\documentclass[11pt]{article}
\usepackage[letterpaper]{geometry}
\usepackage{amsmath,amsthm,amsfonts,amssymb}
\usepackage{enumerate,color,xcolor}
\usepackage{graphicx}
\usepackage{subfigure}
\usepackage[square,numbers]{natbib}
\usepackage{ulem}
\usepackage{cancel}
\usepackage{hyperref}
 
\numberwithin{equation}{section}
\theoremstyle{plain}
\usepackage{filecontents}
\usepackage{silence}
\newtheorem{theorem}{Theorem}

\newtheorem{corollary}[theorem]{Corollary}

\newtheorem{example}[theorem]{Example}
\newtheorem{lemma}[theorem]{Lemma}

\newtheorem{proposition}[theorem]{Proposition}

\newtheorem{remark}[theorem]{Remark}

\DeclareMathOperator*{\argmin}{arg\,min}

\begin{document}

\begin{center}
  \Large \bf Optimal Investment in a Dual Risk Model\end{center}

\author{}
\begin{center}
{Arash Fahim}\,\footnote{Department of Mathematics, Florida State University, 1017 Academic Way, Tallahassee, FL-32306, 
United States of America; Email: arash@math.fsu.edu;},
  Lingjiong Zhu\,\footnote{Corresponding Author. Department of Mathematics, Florida State University, 1017 Academic Way, Tallahassee, FL-32306, United States of America; Email: zhu@math.fsu.edu.
  }
\end{center}

\begin{center}
 \today
\end{center}

%
%

\begin{abstract}
Dual risk models are popular for modeling
a venture capital or high tech company, for which
the running cost is deterministic and the profits
arrive stochastically over time. 
Most of the existing literature on dual risk models
concentrated on the optimal dividend strategies.
In this paper, we propose to study the optimal investment
strategy on research and development for the dual risk models
to minimize the ruin probability of the underlying company. 
We will also study the optimization problem when in addition the investment
in a risky asset is allowed.
\end{abstract}


\section{Introduction}

The classical Cram\'{e}r-Lundberg model, or the classical compound Poisson risk model assumes
that the surplus process of an insurance company follows the dynamics:
\begin{equation}\label{classical}
dX_{t}=\rho dt-dJ_{t},\qquad X_{0}=x>0,
\end{equation}
where $\rho>0$ is the premium rate and $J_{t}=\sum_{i=1}^{N_{t}}Y_{i}$ is a compound Poisson process,
where $N_{t}$ is a Poisson process with intensity $\lambda>0$ and claim sizes $Y_{i}$ are i.i.d. positive random variables
independent of the Poisson process with $\mathbb{E}[Y_{1}]<\infty$. One central question in the ruin theory is to study 
the ruin probability $\mathbb{P}(\tau<\infty)$, where $\tau:=\inf\{t>0:X_{t}<0\}$.

In recent years, there have been a lot of studies in the insurance and finance literature
on the so-called dual risk model, see e.g. \cite{Afonso, Avanzi, AvanziII, BE,CheungI,CheungII,Ng,NgII, RCE,YS}, 
with wealth process following the dynamics:
\begin{equation}\label{dual}
dX_{t}=-\rho dt+dJ_{t},\qquad X_{0}=x>0,
\end{equation}
where $\rho>0$ is the cost of running the company and $J_{t}=\sum_{i=1}^{N_{t}}Y_{i}$, 
is the stream of profits, where $N_{t}$
is a Poisson process with intensity $\lambda>0$
and $Y_{i}$ are i.i.d. $\mathbb{R}^{+}$ valued random variables
with common probability density function $p(y)$, $y>0$, independent of the Poisson process. 
The dual risk model is used to model the wealth of a venture capital,
whose profits depend on the research and development. 
The classical risk model \eqref{classical} is most often interpreted as the surplus of an insurance company.
On the other hand, the dual risk model \eqref{dual} can be understood as the wealth of a venture capital
or high tech company. The analogue of the premium in the classical model is the running cost
in the dual model, and the claims become the future profits of the company.
The ruin probability and the Laplace transform of the ruin time have been well
studied for the dual risk model; see e.g. Afonso et al. \cite{Afonso}.
When there is a random delay
for the innovations turned to profits, the dual risk model
becomes time inhomogeneous and the ruin probabilities and the distribution of the ruin times
are studied in \cite{ZhuDelayed}. 

One of the most fundamental questions in the dual risk model 
is the optimal dividend strategy.
Avanzi et al. \cite{Avanzi} worked on optimal dividends in the dual risk model where
the optimal strategy is a barrier strategy. Avanzi et al. \cite{AvanziII} studied a dividend barrier strategy for the dual risk model 
whereby dividend decisions are made only periodically, but still allow ruin to occur at any time.
A dual model with a threshold dividend strategy, with exponential interclaim times
was studied in Ng \cite{Ng}.
Afonso et al. \cite{Afonso} also worked on dividend problem in the dual risk model, assuming exponential interclaim times.
A new approach for the calculation of expected discounted dividends was presented
and ruin and dividend probabilities, number of dividends, time to a dividend, 
and the distribution for the amount of single dividends were studied.
Dividend moments in the dual risk model were considered in Cheung and Drekic \cite{CheungII}. They derived 
integro-differential equations for the moments of the total discounted dividends which can be solved explicitly 
assuming the jump size distribution has a rational Laplace transform.
The expected discounted dividends assuming
the profits follow a Phase Type distribution were studied in Rodr\'{i}guez et al. \cite{RCE} . 
The Laplace transform of the ruin
time, expected discounted dividends for the Sparre-Andersen dual model
were derived in Yang and Sendova \cite{YS}.
More recently, \cite{Yang2020} obtained an explicit expression
of the expected discounted discounted dividends in a dual risk model
with the threshold dividend strategy
and the optimal threshold level was derived.
\cite{Avanzi2020} considered the optimal periodic dividend strategies
for a general class of dual risk models with fixed transaction costs.
In \cite{FZII}, they obtained the asymptotic analysis for optimal dividends 
in the dual risk model.
\cite{Liu2023} studied the optimal dividend strategy for the dual model
with surplus-dependent expense.

So far the optimization problems studied in the literature on dual risk models are
almost exclusively devoted to the optimal dividend strategy. 
In this paper, we consider a different type of optimization problem. 
For a venture capital, or a high tech company, the investment strategy
on research and development (R\&D) is crucial. A decision to increase
the investment on research and development will increase the running cost
of the company, but that will also boost the possibility of the future profits. 
Therefore, we believe that it is of fundamental interest to understand
the optimal investment strategy to strengthen the position of the company.

It is well known that research and development is a basic engine of economic
and social growth. It is a considerable amount of spending among many leading
corporations in the world. A 2014 FORTUNE article listed the top ten biggest
R\&D spenders worldwide in the year 2013, including Volkswagen, Samsung, 
Intel, Microsoft, Roche, Novartis, Toyota, Johnson \& Johnson, Google and Merck, 
with Intel spent as much as 20.1\% of their revenue on R\&D, see \cite{FORTUNE}. 
Many technology giants increase their R\&D spending consistently, year over year, 
see e.g. Table~\ref{MainTable} for the R\&D and percentage of the revenues of Alphabet, 
Amazon, Tesla in the years 2018-2021\footnote{\label{note1}Available at \texttt{https://www.macrotrends.net/}}.
Notice that in the case of Alphabet, even though the R\&D expenditure increases year by year,
it increases in line with the increase of the total revenues so that as the percentage of revenues,
the  number does not change much. The same can be said about Amazon.
For some companies, both the absolute R\&D expenditure amount and the percentage as the revenues
remain reasonably stable, see e.g. Table~\ref{MainTable} for Merck in the years 2018-2021, 
with the year of 2020 the only exception which witnessed an unusually high R\&D expenditure\footnotemark[3].
For some companies, both the absolute R\&D expenditure amount
and the revenues can change dramatically, 
see e.g. Table~\ref{MainTable} for Alphabet, Amazon, Tesla in the years 2018-2021\footnotemark[3].
The case of Tesla is exceptional but not unusual for a new high-tech company in the sense
that the total revenues has astronomical growth and the R\&D expenditure
as the percentage of revenues actually declines during this period even though
it had a spectacular increase in R\&D expenditure in the year of 2021. 
Another company that has enjoyed similar phenomenal growth as Tesla is the Amazon, see Table~\ref{MainTable}.
But Amazon's overall growth is not as fast as Tesla.

Since it is expensed rather than capitalized, cuts on research
and development increases in profit in the short term, but it can hurt the strength of a company
in the long run, even if the detrimental impact of the cuts may not be felt for a few years. 
In the most recent recession, firms with revenues greater than 100 million USD reduced
their research and development intensity (divided by revenue) by 5.6\%, even though
the advertising intensity actually increased 3.4\%, see \cite{HBS}. In the long run, 
the research and development does help the company grow and increase the value
of a company. Using a measure of the so-called research quotient, a study over
all publicly traded US companies from 1981 through 2006 suggested
that a 10\% increase in research quotient, results an increase in market value of 1.1\%, see \cite{HBS}.
Indeed, the US government also encourages the research and development activities.
The Research \& Experimentation Tax Credit, is a general business tax credit passed
by the Congress in 1981, as a response to the concerns that research spending declines had adversely 
affected the country's economic growth, productivity gains, and competitiveness within the global marketplace.
According to a study by Ernst \& Young, in the year 2005, 17,700 US corporations claimed 6.6 billion USD R\&D tax credits
on their tax returns\footnote{See Supporting innovation and economic growth: The broad impact of the R\&D credit in 2005.
Prepared by Ernst \& Young LLP for the R\&D Coalition. April 2008. 
Available at \texttt{https://www.scribd.com/document/207312025/e-y-RatiosR-DTaxCreditStudy2008final}}.

\begin{table}[htb]
\centering 
\begin{tabular}{|c|c|c|c|c|} 
\hline 
\hline
Alphabet & 2018 & 2019 & 2020 & 2021 
\\
\hline
R\&D (millions) & \$21,419 & \$26,018 & \$27,573 & \$31,562
\\
Revenues (millions) & \$136,819 & \$161,857 & \$182,527 & \$257,637
\\
As \% of Revenues & 15.7\% & 16.1\% & 15.1\% & 12.3\%
\\
\hline 
\hline 
Amazon & 2018 & 2019 & 2020 & 2021 
\\
\hline
R\&D (millions) & \$28,837 & \$35,931 & \$42,740 & \$56,052
\\
Revenues (millions) & \$232,887 & \$280,522 & \$386,064 & \$469,822
\\
As \% of Revenues & 12.4\% & 12.8\% & 11.1\% & 11.9\%

\\
\hline 
\hline 
Tesla & 2018 & 2019 & 2020 & 2021
\\
\hline
R\&D (millions) & \$1,460 & \$1,343 & \$1,491 & \$2,593
\\
Revenues (millions) & \$21,461 & \$24,578 & \$31,536 & \$58,823
\\
As \% of Revenues & 6.8\% & 5.5\% & 4.7\% & 4.4\%
\\
\hline 
\hline 
Merck & 2018 & 2019 & 2020 & 2021 
\\
\hline
R\&D (millions) & \$9,752 & \$9,724 & \$13,397 & \$12,245
\\
Revenues (millions) & \$42,294 & \$39,121 & \$41,518 & \$48,704
\\
As \% of Revenues & 23.1\% & 24.9\% & 32.3\% & 25.1\%
\\
\hline 
\end{tabular}
\caption{R\&D spending by Alphabet, Amazon, Tesla and Merck during 2018-2021.}
\label{MainTable} 
\end{table}

Optimal investment problems have a long history in finance and related fields.
For example, \cite{Merton,Merton1971} formulated and studied the problem
of optimal allocation between risky assets and a risk-free asset to maximize expected utility;
\cite{FZ1991} considered the optimal investment and consumption problem where short-selling is not allowed
but borrowing is allowed. 
\cite{Davis1990,SS1994} studied optimal investment and consumption with proportional transaction costs
and \cite{Morton1995} considered optimal portfolio management with fixed transaction costs.
\cite{GZ1993} studied optimal investment strategies for controlling drawdowns.
\cite{FlemingSheu} studied the optimal investment problem to maximize the long-term growth rate
of expected utility of wealth.
\cite{Hipp} studied the optimal investment for insurers.
\cite{Carr2001} considered the problem of optimal investment in a risky asset,
and in derivatives written on the price process of this asset.
Finally, there are also a limited number
of works on the optimal venture capital investments, see e.g. \cite{BE}.
However, to the best of our knowledge, the optimal investment in research and development
for the dual risk model has never been studied in the previous literature, and our paper is the first one
that considers this problem.

We propose to study the optimal investment
strategy on research and development for the dual risk models
to minimize the ruin probability of the underlying company. 
In addition to the investment in research and development, we will
also allow the investment in a risky asset, e.g. a market index. 
The possibility that an insurer can invest part of the surplus into a risky asset
to minimize the ruin probability has been studied by Browne \cite{Browne}
for the case that the insurance business is modeled by a Brownian motion with constant drift
and the risky asset is modeled as a geometric Brownian motion.
Later, Hipp and Plum \cite{Hipp} studied the optimal investment in a market index
for insurers in the classical compound Poisson risk model. 
We will study the the optimal investment problem when both investment
in research and development and investment in a risky asset are allowed.
Unlike the problem of minimizing
the ruin probability for an insurer in the classical risk model \cite{Hipp},
we will obtain closed-form formulas in the dual risk model. 

Since the works of Browne \cite{Browne} and Hipp and Plum \cite{Hipp}, 
the optimal investment in the market for the classical risk model
and related models have been extensively studied. 
In Liu and Yang \cite{Liu}, they generalized
the works by Hipp and Plum \cite{Hipp} by including a risk-free asset.
In Schmidli \cite{Schmidli}, the optimization problem
of minimizing the ruin probability for the classical risk model
is studied when investment in a risky assent and proportional reinsurance
are both allowed. The asymptotic ruin probability
for the classical risk model under the optimal investment
in a risky asset is obtained by Gaier et al. \cite{Gaier} for large initial wealth.
The asymptotics for small claim sizes were obtained in 
Hipp \cite{Hipp2004}.
In Yang and Zhang \cite{YZ}, they studied the optimal investment
for an insurer when
the risk process is compound Poisson process perturbed by a standard Brownian motion 
and the insurer can invest in the money market and in a risky asset. 
In Gaier and Grandits \cite{GG}, the case when the claim sizes
are of regularly varying tails were studied. The results
were then extended to include interest rates in \cite{GG2}.
The case for subexponential claims was investigated in Schmidli \cite{Schmidli2}.
In Promislow and Young \cite{Promislow}, they studied the problem of
minimizing the probability of ruin of an insurer when the claim process is modeled by 
a Brownian motion with drift optimizing over the investment in a risky asset 
and purchasing quota-share reinsurance.
In Wang et al. \cite{Wang}, 
they adopted the martingale approach to study 
the optimal investment problem for an insurer when the insurer's risk process is modeled by a L\'{e}vy process 
with possible investment in a security market described by the standard Black-Scholes model.
When the underlying investor is an individual rather than an insurance company, 
the optimal investment problem of minimizing the ruin probability
was studied in e.g. Bayraktar and Young \cite{BY}.
In Azcue and Muler \cite{Azcue}, they studied the minimization of the ruin probability
for the classical risk model with possible investment in a risky asset
that follows a geometric Brownian motion under the borrowing constraints.
There have been many other works in this area. For a survey, we refer to
Paulsen \cite{Paulsen} and the references therein.

This paper is organized as follows.
We first introduce a state-dependent dual risk model that generalizes the classical dual risk model (Section~\ref{sec:state}).
When the size of a company increases,
the cost usually also increases, while the resource of income will also increase in general, 
which makes it natural to study a state-dependent dual risk model.
Then we study the optimal investment strategy on research and development
to minimize the ruin probability of the company (Section~\ref{sec:min}), 
with a further discussion of a state-dependent example in Section~\ref{sec:example}. 
As a special case, the state-independent model is discussed in Section~\ref{sec:independent},
with a further discussion of a state-independent example in Section~\ref{sec:indep:example}.
Next, we study the joint investment in research and development and a market index 
to minimize the ruin probability in Section~\ref{sec:market}.
Finally, we provide some numerical studies in Section~\ref{sec:numerical} to better understand how the minimized ruin probability 
and the optimal strategy depend on the parameters in the model.

\section{A State-Dependent Dual Risk Model}\label{sec:state}

We introduce a state-dependent dual risk model with the wealth process
being defined as follows:
\begin{equation}\label{state:dependent:X}
dX_{t}=-\rho(X_{t})dt+dJ_{t},\qquad X_{0}>0,
\end{equation}
where $J_{t}=\sum_{i=1}^{N_{t}}Y_{i}$, where $N_{t}$ is a simple point process with
intensity $\lambda(X_{t-})$ at time $t$,
and $Y_{i}$ are i.i.d. positive random variables with finite mean
and independent of $\mathcal{F}_{\tau_{i}-}$, where $\mathcal{F}_{t}$
is the natural filtration generated by $X_{t}$ process, $\tau_{i}$
is the $i$-th arrival time of $N_{t}$ and
we further assume that $\rho(\cdot),\lambda(\cdot):\mathbb{R}_{+}\rightarrow\mathbb{R}_{+}$
are increasing functions. The state-dependent dual risk model \eqref{state:dependent:X}
was first introduced in \cite{ZhuState}, in which ruin probability and the Laplace
transform of the ruin time were studied.

The motivation of introducing state-dependence for the dual risk model
is the following. First, the cost of a company usually increases as the size of the company increases. 
For example, the running cost of a small business and a Fortune 500 company
are vastly different. Second, as the size of a company increases, the arrival intensity
of the future profits might increase. It may be due to the fact that the larger a company gets,
the more resources for income it will get. It is also well known in the finance literature
that as a company gets larger and stronger, it can enjoy more benefits, e.g. net present value (NPV),
which for example might be due to the opportunities brought by franchising. 
As we can see from Table~\ref{MainTable}, 
the R\&D expenditure may be far from being constant as the size of the company
and the revenue of the company change. More realistically, the R\&D expenditure
and other costs of running the company should be state-dependent.

Let $\tau:=\inf\{t>0: X_{t}\leq 0\}$ be the ruin time of $X_{t}$ process. The eventual ruin probability
is defined as the function $\psi(x):=\mathbb{P}(\tau<\infty|X_{0}=x)$ to emphasize the dependence 
on the initial wealth $x$.
Note that for the state-independent dual risk model, $\lambda(\cdot)\equiv\lambda$
and $\rho(\cdot)\equiv\rho$, 
under the assumption $\lambda\mathbb{E}[Y_{1}]>\rho$, the ruin 
probability $\psi(x)$ is less than $1$. Indeed, $\psi(x)=e^{-\alpha x}$, where
$\alpha>0$ is the unique solution to the equation; see e.g. Afonso et al. \cite{Afonso}:
\begin{equation}\label{alphaEqn}
\rho\alpha+\lambda\int_{0}^{\infty}[e^{-\alpha y}-1]p(y)dy=0.
\end{equation}

For the state-dependent dual risk model, there is no simple closed-form
formula for the ruin probability. Nevertheless, for the special case when
the jump sizes $Y_{i}$ are i.i.d. exponentially distributed,
there is a closed-form expression for the ruin probability; see Theorem~1 in \cite{ZhuState}.

Finally, we notice that the $X_{t}$ process in \eqref{state:dependent:X} is an extension
of the (nonlinear) marked Hawkes process with exponential kernel (see e.g. \cite{Hawkes,Bremaud,GZI,GZII,ZhuI}), 
that is, $N_{t}$ is a simple point process with intensity $\lambda(X_{t})$, where
\begin{equation}\label{intensity:dynamics}
X_{t}:=X_{0}e^{-\beta t}+\sum\nolimits_{i:\tau_{i}<t}Y_{i}e^{-\beta(t-\tau_{i})},
\end{equation}
where $\tau_{i}$ is the $i$-th arrival time of $N_{t}$,  
and $Y_{i}$ are i.i.d. positive random variables independent of $\mathcal{F}_{\tau_{i}-}$
with finite mean and $X_{0},\beta>0$ are given constants, 
where $X_{t}$ in \eqref{intensity:dynamics} satisfies the dynamics \eqref{state:dependent:X} with $\rho(x):=\beta x$. 
When $\lambda(\cdot)$ is linear, it is called linear Hawkes process, named after Hawkes \cite{Hawkes}.
When $\lambda(\cdot)$ is nonlinear, the Hawkes process is said to be nonlinear which
was first introduced by Br\'{e}maud and Massouli\'{e} \cite{Bremaud}.
Hawkes processes have wide applications in finance, neuroscience, social networks, criminology, seismology, and many other fields;
see \cite{GaoFZhu2} and the references therein.
Since the $X_{t}$ process in \eqref{state:dependent:X} is an extension
of the (nonlinear) marked Hawkes process with exponential kernel, our paper
also contributes to the literature of the Hawkes process.

\section{Minimizing the Ruin Probability}\label{sec:min}

In this section, we study the optimization control problem of minimizing the ruin probability
for the dual risk model.
The management of the underlying company can decide whether or not to 
increase the capital spending on research and development to boost the future profits. 
Our goal is to find the optimal expenditure on research and development to minimize the 
probability that the company is eventually ruined. 

Before we proceed, we introduce the investment on research and development $C\in\mathcal{C}$, 
where $\mathcal{C}$ is the set of all admissible strategies, defined as
\begin{align}
\mathcal{C}:=\left\{C:[0,\infty)\times\Omega\rightarrow\mathbb{R}_{\geq 0}:\text{$C$ is progressively measurable, bounded and predictable}\right\}.
\end{align}
Given the control $C\in\mathcal{C}$, the wealth process has the dynamics
\begin{equation}
dX_{t}^{C}=-(\rho(X_{t})+C_{t})dt+dJ_{t}^{C},
\end{equation}
where $\rho:\mathbb{R}_{+}\rightarrow\mathbb{R}_{+}$ is increasing
and $J_{t}=\sum_{i=1}^{N_{t}}Y_{i}$, where $Y_{i}$ are defined
same as before and $N_{t}$ is a simple point process with
intensity $F(X_{t-},C_{t-})$ at time $t$,
where $F(x,c):\mathbb{R}_{+}\times\mathbb{R}_{+}\rightarrow\mathbb{R}_{+}$ is measurable in $(x,c)$ and increasing in both $x$ and $c$
and $F(x,0)=\lambda(x)$ for every $x\in\mathbb{R}_{+}$, where $\lambda:\mathbb{R}_{+}\rightarrow\mathbb{R}_{+}$ is increasing.

We define $\tau^{C}$ as the ruin time of the $X^{C}$
process under the control $C\in\mathcal{C}$ 
by $\tau^{C}:=\inf\{t\geq 0: X^{C}_{t}\leq 0\}$.
We are interested in studying the optimization problem:
\begin{equation}
V(x):=\min_{C\in\mathcal{C}}\mathbb{P}\left(\tau^{C}<\infty|X_{0}^{C}=x\right).
\end{equation}
From the optimal control point of view, it is also interesting to study the state-dependent case, 
which adds a technical contribution to the literature of stochastic optimal control theory.
We will show that the optimal strategy is in general state-dependent
when the underlying dual risk model is state-dependent, and it exhibits
a closed-form expression.

\begin{theorem}\label{MainThm}
The optimal strategy $C^{\ast}$ is given by
\begin{equation}\label{arg:min:eqn}
C^{\ast}_{t}=C^{\ast}(X_{t})\in\argmin_{C\geq 0}\frac{\rho(X_{t})+C}{F(X_{t},C)},
\end{equation}
provided that the minimum exists.
\end{theorem}

\begin{proof}[Proof of Theorem~\ref{MainThm}]
For any control $C\in\mathcal{C}$, we have
\begin{equation}\label{X:controlled:eqn}
dX^{C}_{t}=-(\rho(X_{t})+C_{t})dt+dJ^{C}_{t},
\end{equation}
where $J^{C}_{t}=\sum_{i=1}^{N^{C}_{t}}Y_{i}$, where $N^{C}_{t}$
is a simple point process with intensity $F(X_{t-},C_{t-})$ at time $t$
and $Y_{i}$ are i.i.d. with probability density function $p(y)$ defined as before.

Let us introduce a random time change and define the random time $T(t)$ via:
\begin{equation}\label{eqn:random:time}
\int_{0}^{T(t)}F(X_{s-},C_{s-})ds=t.
\end{equation}
Then, it is easy to see that $T(0)=0$ and $T(t)\rightarrow\infty$ as $t\rightarrow\infty$ since $C\in\mathcal{C}$
is bounded. It follows from \eqref{X:controlled:eqn} that
\begin{equation}
dX_{T(t)}=-(\rho(X_{T(t)})+C_{T(t)})dT(t)+dJ^{C}_{T(t)}.
\end{equation}
Under the random time change \eqref{eqn:random:time}, we have
\begin{equation*}
\frac{dT(t)}{dt}=\frac{1}{F(X_{t},C_{t})},
\end{equation*}
and $J^{C}_{T(t)}$ is distributed as $\overline{J}_{t}:=\sum_{i=1}^{\overline{N}_{t}}Y_{i}$,
where $\overline{N}_{t}$ is a standard Poisson process with intensity $1$; 
see e.g. Meyer \cite{Meyer} for the random time change for simple point processes.
Therefore, we obtain
\begin{equation}
dX_{T(t)}=-\frac{\rho(X_{T(t)})+C_{T(t)}}{F(X_{t},C_{t})}dt+d\overline{J}_{t}.
\end{equation}
Let us also notice that
$\mathbb{P}(\text{$X_{t}$ ever gets ruined})
=\mathbb{P}(\text{$X_{T(t)}$ ever gets ruined})$.
Therefore, the optimal strategy is given by \eqref{arg:min:eqn}
provided that the minimum exists.
This completes the proof.
\end{proof}

In Theorem~\ref{MainThm}, we obtain the closed-form
expression of the optimal strategy $C^{\ast}$. However, we
do not have a closed-form for the minimized ruin probability
$\mathbb{P}(\tau^{C^{\ast}}<\infty|X_{0}^{C^{\ast}}=x)$.
Next, we will show that we can obtain a closed-form
for the ruin probability in the special case
when the jump sizes $Y_{i}$ follow exponential distributions.
We first recall the following result from \cite{ZhuState}, 
which states that the ruin probability for a state-dependent dual
risk model with the exponentially distributed $Y_{i}$ has a closed-form expression.

\begin{theorem}[Theorem~1 in \cite{ZhuState}]\label{MainThm2}
Consider the dual risk model:
$dX_{t}=-\rho(X_{t})dt+dJ_{t}$,
where $X_{0}=x>0$, $J_{t}=\sum_{i=1}^{N_{t}}Y_{i}$, where $Y_{i}$ are
exponential random variables with
the probability density function $p(y)=\nu e^{-\nu y}$, $\nu>0$,
and $N_{t}$ is a simple point process with
intensity $\lambda(X_{t-})$ at time $t$,
where $\rho(\cdot),\lambda(\cdot):\mathbb{R}_{+}\rightarrow\mathbb{R}_{+}$
are increasing functions.
Then,
\begin{equation}
\mathbb{P}\left(\tau<\infty|X_{0}=x\right)
=\frac{\int_{x}^{\infty}\frac{\lambda(y)}{\rho(y)}
e^{\nu y-\int_{0}^{y}\frac{\lambda(w)}{\rho(w)}dw}dy}
{\int_{0}^{\infty}\frac{\lambda(y)}{\rho(y)}
e^{\nu y-\int_{0}^{y}\frac{\lambda(w)}{\rho(w)}dw}dy}.
\end{equation}
\end{theorem}

As a corollary of Theorem~\ref{MainThm} and Theorem~\ref{MainThm2},
we obtain the closed-form for the minimized ruin probability
when the jump sizes $Y_{i}$ are i.i.d. exponentially distributed.

\begin{proposition}\label{explicitProb}
Assume $p(y)=\nu e^{-\nu y}$, where $\nu>0$. 
Also assume that the integral $\int_{0}^{\infty}\frac{F(y,C^{\ast}(y))}{\rho(y)+C^{\ast}(y)}
e^{\nu y-\int_{0}^{y}\frac{F(w,C^{\ast}(w))}{\rho(w)+C^{\ast}(w)}dw}dy$ exists
and is finite. Then,
\begin{equation}
\min_{C\in\mathcal{C}}\mathbb{P}\left(\tau^{C}<\infty|X_{0}^{C}=x\right)
=\frac{\int_{x}^{\infty}\frac{F(y,C^{\ast}(y))}{\rho(y)+C^{\ast}(y)}
e^{\nu y-\int_{0}^{y}\frac{F(w,C^{\ast}(w))}{\rho(w)+C^{\ast}(w)}dw}dy}
{\int_{0}^{\infty}\frac{F(y,C^{\ast}(y))}{\rho(y)+C^{\ast}(y)}
e^{\nu y-\int_{0}^{y}\frac{F(w,C^{\ast}(w))}{\rho(w)+C^{\ast}(w)}dw}dy}.
\end{equation}
\end{proposition}

\begin{proof}[Proof of Proposition~\ref{explicitProb}]
The proposition follows immediately from Theorem~\ref{MainThm} and Theorem~\ref{MainThm2}. 
\end{proof}

\subsection{A State-Dependent Example}\label{sec:example}

In this section, we study a state-dependent example
in details. We assume that
\begin{equation}
F(x,c)=\lambda(x)+\delta(x)c^{\gamma},
\end{equation}
where $\delta(\cdot):\mathbb{R}_{+}\rightarrow\mathbb{R}_{+}$
is increasing, and $\gamma>0$.
We recall that $\lambda(\cdot)$ is increasing 
and thus $\lambda(\cdot)\geq\lambda(0)>0$.
Let us also assume that
$\rho(\cdot)\leq\rho(\infty)<\infty$. 
Under our assumptions, $F(x,c)$ is increasing in both $x$
and $c$, and $F(x,0)=\lambda(x)$.

Notice when $\gamma>1$, for any constant strategy $C_{t}\equiv C$, where $C>0$ is sufficiently large, 
the ruin probability is bounded above by 
the ruin probability of the following process:
\begin{equation}\label{infty:process}
dX_{t}=-(\rho(\infty)+C)dt+dJ_{t},
\end{equation}
where $J_{t}=\sum_{i=1}^{N_{t}}Y_{i}$
is compound Poisson with 
$N_{t}$ being the Poisson process with intensity
$\lambda(0)+\delta(0)C^{\gamma}$.

By the ruin probability for state-independent dual risk model (see e.g. Afonso \cite{Afonso}), 
the ruin probability of the $X_{t}$ process defined in \eqref{infty:process} is given by $e^{-\alpha_{C}x}$, where $\alpha_{C}$ is the unique positive solution
to the equation:
\begin{equation}
(\rho(\infty)+C)\alpha_{C}+(\lambda(0)+\delta(0)C^{\gamma})\int_{0}^{\infty}[e^{-\alpha_{C}y}-1]p(y)dy=0.
\end{equation}
We can rewrite this equation as:
\begin{equation}
\frac{\rho(\infty)+C}{\lambda(0)+\delta(0)C^{\gamma}}\alpha_{C}=\int_{0}^{\infty}[1-e^{-\alpha_{C}y}]p(y)dy.
\end{equation}
The right hand side of the above equation is bounded between $0$ and $1$.
In the left hand side of the above equation, $\lim_{C\rightarrow\infty}\frac{\rho(\infty)+C}{\delta(0)C^{\gamma}}=0$,
which implies that $\alpha_{C}\rightarrow\infty$ as $C\rightarrow\infty$. 
Hence, $V(x)\leq\inf_{C>0}e^{-\alpha_{C}x}=0$ and the minimized ruin probability
is trivially zero.

Therefore, in the rest of this section, we only consider two cases: (i) $0<\gamma<1$; (ii) $\gamma=1$.

\subsubsection{The \texorpdfstring{$0<\gamma<1$}{0<gamma<1} Case}

Under the assumption that $0<\gamma<1$, it is easy to see from Theorem~\ref{MainThm} that
the optimal strategy $C_{T(t)}$ is the strategy that minimizes
the drift:
\begin{equation}\label{above:eqn}
\frac{\rho(X_{T(t)})+C_{T(t)}}{\lambda(X_{T(t)})+\delta(X_{T(t)})C_{T(t)}^{\gamma}}.
\end{equation}
It is easy to compute from \eqref{above:eqn} that the optimal strategy satisfies
\begin{equation}
\lambda(X_{T(t)})+\delta(X_{T(t)})(1-\gamma)C_{T(t)}^{\gamma}=\rho(X_{T(t)})\delta(X_{T(t)})\gamma C_{T(t)}^{\gamma-1}.
\end{equation}
Therefore, for any $t>0$, the optimal strategy $C_{t}$ satisfies
\begin{equation}
\lambda(X_{t})+\delta(X_{t})(1-\gamma)C_{t}^{\gamma}=\rho(X_{t})\delta(X_{t})\gamma C_{t}^{\gamma-1}.
\end{equation}
It is clear that the optimal strategy $C_{t}$ is a function of $X_{t}$ and we denote it as $C^{\ast}(X_{t})$.
Then under the optimal strategy,
\begin{equation}
dX_{t}=-(\rho(X_{t})+C^{\ast}(X_{t}))dt+dJ_{t},
\end{equation}
where $J_{t}=\sum_{i=1}^{N_{t}}Y_{i}$, where $N_{t}$ has intensity
$\lambda(X_{t-})+\delta(X_{t-})C^{\ast}(X_{t-})^{\gamma}$ at time $t$.

When the probability density function $p(y)=\nu e^{-\nu y}$ of jump sizes $Y_{i}$ is exponential, 
it follows from Proposition~\ref{explicitProb} that we have the following result:

\begin{proposition}\label{explicitProb:exp}
Assume $p(y)=\nu e^{-\nu y}$, where $\nu>0$. 
Also assume that the integral $\int_{0}^{\infty}\frac{\lambda(y)+\delta(y)C^{\ast}(y)^{\gamma}}{\rho(y)+C^{\ast}(y)}
e^{\nu y-\int_{0}^{y}\frac{\lambda(w)+\delta(w)C^{\ast}(w)^{\gamma}}{\rho(w)+C^{\ast}(w)}dw}dy$ exists
and is finite. Then,
\begin{equation}\label{eqn:V:x}
V(x)=\frac{\int_{x}^{\infty}\frac{\lambda(y)+\delta(y)C^{\ast}(y)^{\gamma}}{\rho(y)+C^{\ast}(y)}
e^{\nu y-\int_{0}^{y}\frac{\lambda(w)+\delta(w)C^{\ast}(w)^{\gamma}}{\rho(w)+C^{\ast}(w)}dw}dy}
{\int_{0}^{\infty}\frac{\lambda(y)+\delta(y)C^{\ast}(y)^{\gamma}}{\rho(y)+C^{\ast}(y)}
e^{\nu y-\int_{0}^{y}\frac{\lambda(w)+\delta(w)C^{\ast}(w)^{\gamma}}{\rho(w)+C^{\ast}(w)}dw}dy}.
\end{equation}
\end{proposition}

\begin{proof}[Proof of Proposition~\ref{explicitProb:exp}]
The proposition follows immediately from Proposition~\ref{explicitProb}. 
\end{proof}

Next, in the following example, we show that with particular model specifications, 
the optimal $C^{\ast}$ and
the minimized ruin probability $V(x)$ in \eqref{eqn:V:x}
admit a simpler closed-form formulas.

\begin{example}\label{StateExI}
Let $\rho(x)=\rho_{0}$, $\lambda(x)=\lambda_{0}(c_{1}x+c_{2})$, 
and $\delta(x)=\delta_{0}(c_{1}x+c_{2})$, where $\rho_{0},\lambda_{0},\delta_{0},c_{1},c_{2}$
are positive constants. Then, the optimal investment rate $C^{\ast}(x)$ 
is a constant $C^{\ast}(x)\equiv C_{0}$, where $C_{0}$ is the unique positive solution
to the equation:
\begin{equation}
\lambda_{0}+\delta_{0}(1-\gamma)C_{0}^{\gamma}=\rho_{0}\delta_{0}\gamma C_{0}^{\gamma-1}.
\end{equation}
Hence, the minimized ruin probability in \eqref{eqn:V:x} can be computed as:
\begin{align}
V(x)
&=\frac{\int_{x}^{\infty}\frac{\lambda_{0}+\delta_{0}C_{0}^{\gamma}}{\rho_{0}+C_{0}}(c_{1}y+c_{2})
e^{\nu y-\int_{0}^{y}\frac{\lambda_{0}+\delta_{0}C_{0}^{\gamma}}{\rho_{0}+C_{0}}(c_{1}w+c_{2})dw}dy}
{\int_{0}^{\infty}\frac{\lambda_{0}+\delta_{0}C_{0}^{\gamma}}{\rho_{0}+C_{0}}(c_{1}y+c_{2})
e^{\nu y-\int_{0}^{y}\frac{\lambda_{0}+\delta(w)C_{0}^{\gamma}}{\rho_{0}+C_{0}}(c_{1}w+c_{2})dw}dy}
\\
&=\frac{\int_{x}^{\infty}(c_{1}y+c_{2})
e^{\left(\nu-\frac{\lambda_{0}+\delta_{0}C_{0}^{\gamma}}{\rho_{0}+C_{0}}c_{2}\right)y
-\frac{\lambda_{0}+\delta_{0}C_{0}^{\gamma}}{\rho_{0}+C_{0}}\frac{c_{1}}{2}y^{2}}dy}
{\int_{0}^{\infty}(c_{1}y+c_{2})
e^{\left(\nu-\frac{\lambda_{0}+\delta_{0}C_{0}^{\gamma}}{\rho_{0}+C_{0}}c_{2}\right)y
-\frac{\lambda_{0}+\delta_{0}C_{0}^{\gamma}}{\rho_{0}+C_{0}}\frac{c_{1}}{2}y^{2}}dy}
\nonumber
\\
&=
\frac{\frac{1}{4d^{3/2}}e^{-dy^{2}}\left[\sqrt{\pi}e^{\frac{c^{2}}{4d}+dy^{2}}(ac+2bd)
\text{erf}(\frac{2dy-c}{2\sqrt{d}})-2a\sqrt{d}e^{cy}\right]\bigg|_{y=x}^{\infty}}
{\frac{1}{4d^{3/2}}e^{-dy^{2}}\left[\sqrt{\pi}e^{\frac{c^{2}}{4d}+dy^{2}}(ac+2bd)
\text{erf}(\frac{2dy-c}{2\sqrt{d}})-2a\sqrt{d}e^{cy}\right]\bigg|_{y=0}^{\infty}}
\nonumber
\\
&=\frac{2a\sqrt{d}e^{cx-dx^{2}}+\sqrt{\pi}e^{\frac{c^{2}}{4d}}(ac+2bd)
\text{erfc}(\frac{2dx-c}{2\sqrt{d}})}{2a\sqrt{d}+\sqrt{\pi}e^{\frac{c^{2}}{4d}}(ac+2bd)\text{erfc}(\frac{-c}{2\sqrt{d}})},
\nonumber
\end{align}
where $\text{erf}(x):=\frac{2}{\sqrt{2\pi}}\int_{0}^{x}e^{-t^{2}}dt$ is the error function 
and $\text{erfc}(x):=1-\text{erf}(x)$ is the complementary error function
and $a:=c_{1}$, $b:=c_{2}$, and
\begin{equation}
c:=\nu-\frac{\lambda_{0}+\delta_{0}C_{0}^{\gamma}}{\rho_{0}+C_{0}}c_{2},
\qquad
d:=\frac{\lambda_{0}+\delta_{0}C_{0}^{\gamma}}{\rho_{0}+C_{0}}\frac{c_{1}}{2}.
\end{equation}
\end{example}

\subsubsection{The \texorpdfstring{$\gamma=1$}{gamma=1} Case}

When $\gamma=1$, it follows from Theorem~\ref{MainThm} that 
the optimal $C^{\ast}(x)$ satisfies $C^{\ast}(x)=0$ in the region
where $\delta(x)\leq\frac{\lambda(x)}{\rho(x)}$ and the ``optimal'' $C^{\ast}(x)=\infty$ in the region where 
$\delta(x)>\frac{\lambda(x)}{\rho(x)}$.

\begin{remark}
If we impose a research and development budget constraint by $M\in(0,\infty)$, the maximum capacity. 
Then, the admissible set of controls is given by $\mathcal{C}_{M}:=\{C\in\mathcal{C}: \sup_{t\geq 0}C_{t}\leq M\}$. 
Then the above analysis implies that $C^{\ast}(x)=0$ in the region $\delta(x)\leq\frac{\lambda(x)}{\rho(x)}$
and $C^{\ast}(x)=M$ in the region $\delta(x)>\frac{\lambda(x)}{\rho(x)}$.
\end{remark}

Next, in the following example, we show that with particular model specifications, the optimal $C^{\ast}$
the minimized ruin probability $V(x)$ admit simpler closed-form formulas.

\begin{example}\label{StateExII}
Let $\rho(x)=\rho_{0}(c_{1}x+c_{2})$, $\lambda(x)=\left(\nu+\frac{\lambda_{0}}{1+x}\right)\rho(x)$, 
and $\delta(x)=\delta_{0}$, where $\rho_{0},c_{1},c_{2},\lambda_{0},\delta_{0}$ are positive constants.
We further assume that
$\nu<\delta_{0}<\nu+\lambda_{0}$.
Then, the optimal $C^{\ast}$ is given by:
\begin{equation}
C^{\ast}(x)=
\begin{cases}
0 &\text{if $x\leq\frac{\lambda_{0}-\delta_{0}+\nu}{\delta_{0}-\nu}$},
\\
+\infty &\text{if $x>\frac{\lambda_{0}-\delta_{0}+\nu}{\delta_{0}-\nu}$}.
\end{cases}
\end{equation}
Let us define:
\begin{equation}
x^{\ast}:=\frac{\lambda_{0}-\delta_{0}+\nu}{\delta_{0}-\nu}.
\end{equation}
Then, we can compute that for any $y\leq x^{\ast}$,
\begin{equation}
\int_{0}^{y}\frac{\lambda(w)+\delta(w)C^{\ast}(w)}{\rho(w)+C^{\ast}(w)}dw
=\int_{0}^{y}\left(\nu+\frac{\lambda_{0}}{1+w}\right)dw
=\nu y+\lambda_{0}\log(1+y),
\end{equation}
and for any $y>x^{\ast}$,
\begin{equation}
\int_{0}^{y}\frac{\lambda(w)+\delta(w)C^{\ast}(w)}{\rho(w)+C^{\ast}(w)}dw
=\nu x^{\ast}+\lambda_{0}\log(1+x^{\ast})
+\delta_{0}(y-x^{\ast}).
\end{equation}
Therefore, for $x>x^{\ast}$, we have
\begin{align}
&\int_{x}^{\infty}\frac{\lambda(y)+\delta(y)C^{\ast}(y)}{\rho(y)+C^{\ast}(y)}
e^{\nu y-\int_{0}^{y}\frac{\lambda(w)+\delta(w)C^{\ast}(w)}{\rho(w)+C^{\ast}(w)}dw}dy
\\
&=\int_{x}^{\infty}\delta_{0}e^{\nu y-\nu x^{\ast}-\lambda_{0}\log(1+x^{\ast})
-\delta_{0}(y-x^{\ast})}dy
=\frac{e^{-\nu x^{\ast}+\delta_{0}x^{\ast}}}{(1+x^{\ast})^{\lambda_{0}}}
\frac{\delta_{0}}{\delta_{0}-\nu}e^{-(\delta_{0}-\nu)x},
\nonumber
\end{align}
and for $x\leq x^{\ast}$, we have
\begin{align}
&\int_{x}^{\infty}\frac{\lambda(y)+\delta(y)C^{\ast}(y)}{\rho(y)+C^{\ast}(y)}
e^{\nu y-\int_{0}^{y}\frac{\lambda(w)+\delta(w)C^{\ast}(w)}{\rho(w)+C^{\ast}(w)}dw}dy
\\
&=\int_{x}^{x^{\ast}}\left(\nu+\frac{\lambda_{0}}{1+y}\right)e^{\nu y-\nu y-\lambda_{0}\log(1+y)}dy
+\frac{1}{(1+x^{\ast})^{\lambda_{0}}}
\frac{\delta_{0}}{\delta_{0}-\nu}
\nonumber
\\
&=\frac{\nu}{1-\lambda_{0}}\left[(1+x^{\ast})^{-\lambda_{0}+1}-(1+x)^{-\lambda_{0}+1}\right]
+(1+x)^{-\lambda_{0}}-(1+x^{\ast})^{-\lambda_{0}}+\frac{1}{(1+x^{\ast})^{\lambda_{0}}}
\frac{\delta_{0}}{\delta_{0}-\nu}.
\nonumber
\end{align}
Hence, we conclude that for $x>x^{\ast}$, we have
\begin{equation}\label{greaterxast}
V(x)=\frac{\frac{e^{-\nu x^{\ast}+\delta_{0}x^{\ast}}}{(1+x^{\ast})^{\lambda_{0}}}
\frac{\delta_{0}}{\delta_{0}-\nu}e^{-(\delta_{0}-\nu)x}}
{\frac{\nu}{1-\lambda_{0}}\left[(1+x^{\ast})^{-\lambda_{0}+1}-1\right]
+1-(1+x^{\ast})^{-\lambda_{0}}+\frac{1}{(1+x^{\ast})^{\lambda_{0}}}
\frac{\delta_{0}}{\delta_{0}-\nu}},
\end{equation}
and for $x\leq x^{\ast}$, we have
\begin{equation}\label{lessxast}
V(x)=\frac{\frac{\nu}{1-\lambda_{0}}\left[(1+x^{\ast})^{-\lambda_{0}+1}-(1+x)^{-\lambda_{0}+1}\right]
+(1+x)^{-\lambda_{0}}-(1+x^{\ast})^{-\lambda_{0}}+\frac{1}{(1+x^{\ast})^{\lambda_{0}}}
\frac{\delta_{0}}{\delta_{0}-\nu}}
{\frac{\nu}{1-\lambda_{0}}\left[(1+x^{\ast})^{-\lambda_{0}+1}-1\right]
+1-(1+x^{\ast})^{-\lambda_{0}}+\frac{1}{(1+x^{\ast})^{\lambda_{0}}}
\frac{\delta_{0}}{\delta_{0}-\nu}}.
\end{equation}
\end{example}

\subsection{The State-Independent Case}\label{sec:independent}

In this section, we consider the state-independent case, that is, 
\begin{equation}\label{indep1}
\rho(\cdot)\equiv\rho,
\qquad
\lambda(\cdot)\equiv\lambda,
\end{equation}
and
\begin{equation}\label{indep2}
F(\cdot,c)\equiv F(c),
\end{equation}
where $\rho,\lambda>0$ and $F:\mathbb{R}_{+}\rightarrow\mathbb{R}_{+}$ is increasing.
Under the assumptions~\eqref{indep1}, \eqref{indep2}, 
we have the following result
which is a corollary of Theorem~\ref{MainThm}
and the ruin probability
for the state-independent dual risk model (equation~\eqref{alphaEqn}).

\begin{theorem}\label{thm:indep}
The optimal strategy $C^{\ast}$ is constant, given by
\begin{equation}
C^{\ast}=\argmin_{C\geq 0}\frac{\rho+C}{F(C)},
\end{equation}
provided that the minimum exists
and the minimized ruin probability is $V(x)=e^{-\beta x}$, 
where
\begin{equation}\label{beta:eqn:satisfies}
(\rho+C^{\ast})\beta+F(C^{\ast})\int_{0}^{\infty}[e^{-\beta y}-1]p(y)dy=0.
\end{equation}
\end{theorem}

\begin{proof}[Proof of Theorem~\ref{thm:indep}]
Under the assumptions~\eqref{indep1}, \eqref{indep2}, 
it follows from Theorem~\ref{MainThm} that the optimal strategy $C^{\ast}$ is constant, which
is given by
$C^{\ast}=\argmin_{C\geq 0}\frac{\rho+C}{F(C)}$.
With the optimal $C^{\ast}$, we have
\begin{equation}
dX_{t}=-(\rho+C^{\ast})dt+dJ_{t},
\end{equation}
where $J_{t}=\sum_{i=1}^{N_{t}}Y_{i}$
is compound Poisson, where $N_{t}$ is Poisson
with intensity $F(C^{\ast})$.

By the formula for the ruin probability
for the state-independent dual risk model, see e.g. equation~\eqref{alphaEqn}, 
we have $V(x)=e^{-\beta x}$, 
where $\beta$ satisfies the equation~\eqref{beta:eqn:satisfies}.
This completes the proof.
\end{proof}

\subsection{A State-Independent Example}\label{sec:indep:example}

In this section, we consider a state-independent example, that is, 
\begin{equation}
\rho(\cdot)\equiv\rho,
\qquad
\lambda(\cdot)\equiv\lambda,
\end{equation}
and
\begin{equation}
F(x,c)=\lambda+\delta c^{\gamma},
\qquad
\delta, \gamma>0.
\end{equation}
In this special case, by Theorem~\ref{thm:indep}, the optimal strategy $C^{\ast}$ is constant and given by
\begin{equation}
C^{\ast}=\argmin_{C\geq 0}\frac{\rho+C}{\lambda+\delta C^{\gamma}}.
\end{equation}
By following the discussions in the more general state-dependent
case in Section~\ref{sec:example}, the case $\gamma\geq 1$
is trivial and in the rest we only consider the cases $0<\gamma<1$ and
$\gamma=1$.

\subsubsection{The \texorpdfstring{$0<\gamma<1$}{0<gamma<1} Case}

We first consider the case that $0<\gamma<1$. 
In this case, the intensity $F(X_{t},C_{t})=\lambda+\delta C_{t}^{\gamma}$ is a concave and increasing function of $C_{t}$.
What it says is that the initial investment of research and development can boost the
prospect of future profits, but the margin decreases as the increase of the investment. 

When it is allowed to invest in research and development, we will see later,
that the condition 
\begin{equation}\label{ConditionOne}
(\rho-\lambda\mathbb{E}[Y_{1}])-(\delta\gamma)^{\frac{1}{1-\gamma}}\left(\frac{1}{\gamma}-1\right)(\mathbb{E}[Y_{1}])^{\frac{1}{1-\gamma}}<0
\end{equation}
is sufficient to guarantee that $V(x)<1$.
Note that this is weaker than the usual condition $\rho-\lambda\mathbb{E}[Y_{1}]<0$
for the dual risk model.
We have the following result.

\begin{proposition}\label{prop:indep}
Under the assumption \eqref{ConditionOne}, 
\begin{equation}
V(x)=\min_{C\in\mathcal{C}}\mathbb{P}\left(\tau^{C}<\infty|X^{C}_{0}=x\right)=e^{-\beta x},
\end{equation}
where $\beta$ is the unique positive value that satisfies the equation:
\begin{align}\label{BetaEqnI}
&\beta\left[\rho+\left(\frac{1}{\delta\gamma}\right)^{\frac{1}{\gamma-1}}
\left(\frac{\beta}{1-\int_{0}^{\infty}e^{-\beta y}p(y)dy}\right)^{\frac{1}{\gamma-1}}\right]
\\
&\qquad
-\left[\lambda+\delta\left(\frac{1}{\delta\gamma}\right)^{\frac{\gamma}{\gamma-1}}
\left(\frac{\beta}{1-\int_{0}^{\infty}e^{-\beta y}p(y)dy}\right)^{\frac{\gamma}{\gamma-1}}\right]
\left(1-\int_{0}^{\infty}e^{-\beta y}p(y)dy\right)=0,
\nonumber
\end{align}
and the optimal strategy is given by
\begin{equation}
C^{\ast}=\left(\frac{1}{\delta\gamma}\right)^{\frac{1}{\gamma-1}}
\left(\frac{\beta}{1-\int_{0}^{\infty}e^{-\beta y}p(y)dy}\right)^{\frac{1}{\gamma-1}},
\end{equation}
which also satisfies the equation:
\begin{equation}
\lambda+(1-\gamma)\delta(C^{\ast})^{\gamma}=\rho\delta\gamma(C^{\ast})^{\gamma-1}.
\end{equation}
\end{proposition}

\begin{proof}[Proof of Proposition~\ref{prop:indep}]
It follows from Theorem~\ref{thm:indep} that 
the optimal strategy is given by
\begin{equation}
C^{\ast}=\left(\frac{1}{\delta\gamma}\right)^{\frac{1}{\gamma-1}}
\left(\frac{\beta}{1-\int_{0}^{\infty}e^{-\beta y}p(y)dy}\right)^{\frac{1}{\gamma-1}},
\end{equation}
and the minimized ruin probability $V(x)$
satisfies the equation \eqref{BetaEqnI}. 

To show that \eqref{BetaEqnI} has a unique positive solution, 
it is equivalent to show that $F(\beta)=0$ has a unique positive solution 
where 
\begin{equation}
F(\beta):=\beta\left[\rho-(\delta\gamma)^{\frac{1}{1-\gamma}}\left(\frac{1}{\gamma}-1\right)[g(\beta)]^{\frac{1}{1-\gamma}}
-\lambda g(\beta)\right],
\end{equation}
and
\begin{equation}\label{small:g:eqn}
g(\beta):=\frac{1-\int_{0}^{\infty}e^{-\beta y}p(y)dy}{\beta}.
\end{equation}
It is easy to compute that for $\beta>0$,
\begin{equation}
g'(\beta)=\frac{1}{\beta^{2}}\int_{0}^{\infty}\left[\beta ye^{-\beta y}-1+e^{-\beta y}\right]p(y)dy.
\end{equation}
Let $h(x):=xe^{-x}-1+e^{-x}$, $x\geq 0$. Then $h(0)=0$
and $h(x)\rightarrow-1$ as $x\rightarrow\infty$.
Moreover, $h'(x)=-xe^{-x}<0$ for $x>0$. Thus $h(x)\leq 0$ for any $x\geq 0$ 
and therefore, $g'(\beta)\leq 0$ for any $\beta>0$ and $g(\beta)$ is a decreasing function of $\beta$.

Note that $F(\beta)=0$ for $\beta>0$ if and only if
$G(\beta)=0$ for $\beta>0$, where
\begin{equation}\label{Gbeta}
G(\beta):=\rho-(\delta\gamma)^{\frac{1}{1-\gamma}}\left(\frac{1}{\gamma}-1\right)[g(\beta)]^{\frac{1}{1-\gamma}}
-\lambda g(\beta).
\end{equation}
Note that by L'H\^{o}pital's rule,
$\lim_{\beta\rightarrow 0^{+}}g(\beta)=\mathbb{E}[Y_{1}]$.
Therefore,
\begin{equation}
\lim_{\beta\rightarrow 0^{+}}G(\beta)
=(\rho-\lambda\mathbb{E}[Y_{1}])-(\delta\gamma)^{\frac{1}{1-\gamma}}\left(\frac{1}{\gamma}-1\right)(\mathbb{E}[Y_{1}])^{\frac{1}{1-\gamma}}<0.
\end{equation}
On the other hand, $g(\beta)\rightarrow 0$ as $\beta\rightarrow\infty$; therefore
$G(\beta)\rightarrow\rho>0$ as $\beta\rightarrow\infty$.
Since $g(\beta)$ is a decreasing function in $\beta$ and $0<\gamma<1$, 
it follows that $G(\beta)$ is increasing in $\beta$. 
Hence, we conclude that $G(\beta)=0$ has a unique positive solution.
This completes the proof.
\end{proof}

In the following example, we show that 
when $Y_{i}$ are exponentially distributed, 
we are able to compute out $\beta$ and $C^{\ast}$ in simple
closed-forms. 

\begin{example}\label{ExpExample}
When $p(y)=\nu e^{-\nu y}$, $\nu>0$, $\beta$ satisfies
\begin{equation}
\beta\left[\rho+\left(\frac{1}{\delta\gamma}\right)^{\frac{1}{\gamma-1}}(\beta+\nu)^{\frac{1}{\gamma-1}}\right]
=\left[\lambda+\delta\left(\frac{1}{\delta\gamma}\right)^{\frac{\gamma}{\gamma-1}}(\beta+\nu)^{\frac{\gamma}{\gamma-1}}\right]
\frac{\beta}{\beta+\nu},
\end{equation}
which implies that
\begin{equation}
\rho(\beta+\nu)=\lambda+\left(\frac{1}{\gamma}-1\right)
\left(\frac{1}{\delta\gamma}\right)^{\frac{1}{\gamma-1}}(\beta+\nu)^{\frac{\gamma}{\gamma-1}}.
\end{equation}
In particular, when $\gamma=\frac{1}{2}$, we get
$\rho(\beta+\nu)^{2}=\lambda(\beta+\nu)+
\frac{\delta^{2}}{4}$,
which implies
$\beta=\frac{\lambda+\sqrt{\lambda^{2}+\rho\delta^{2}}}{2\rho}-\nu$,
and thus the optimal $C^{\ast}$ is given by
\begin{equation}
C^{\ast}=\frac{\delta^{2}\rho^{2}}{(\lambda+\sqrt{\lambda^{2}+\rho\delta^{2}})^{2}}.
\end{equation}
\end{example}


\begin{remark}\label{ParametersI}
We have already showed in Proposition~\ref{prop:indep} that $V(x)=e^{-\beta x}$, where $\beta$ 
is the unique positive solution to the equation \eqref{BetaEqnI}
and that it is equivalent
\begin{equation}\label{Gbeta0}
\rho-(\delta\gamma)^{\frac{1}{1-\gamma}}\left(\frac{1}{\gamma}-1\right)[g(\beta)]^{\frac{1}{1-\gamma}}
-\lambda g(\beta)=0,
\end{equation}
where $g(\beta)$ is defined in \eqref{small:g:eqn}.
Now, let us discuss how the value $\beta$ (and hence the value function $V(x)=e^{-\beta x}$)
and the optimal investment rate $C^{\ast}$ depend on the parameters
$\rho$, $\lambda$ and $\delta$.
By \eqref{Gbeta0}, we
have the following observations:

(i) As $\rho$ increases, $g(\beta)$ increases.
Since $g(\beta)$ is decreasing in $\beta$, we conclude that $\beta$
decreases as $\rho$ increases. Intuitively it says that as the fixed running
cost for research and investment increases, the ruin probability increases.
Asymptotically, as $\rho\rightarrow 0$, $g(\beta)\rightarrow 0$.
When $g(\beta)\rightarrow 0$, since $0<\gamma<1$, 
we must have $[g(\beta)]^{\frac{1}{1-\gamma}}\ll g(\beta)$.
Therefore, by \eqref{Gbeta0}, as $\rho\rightarrow 0$, we have
$g(\beta)\sim\frac{\rho}{\lambda}$. From the definition of $g(\beta)$, we have $g(\beta)\sim\frac{1}{\beta}$
as $\beta\rightarrow\infty$. Hence, we conclude that
$\beta\sim\frac{\lambda}{\rho}$,
as $\rho\rightarrow 0$.
Therefore, the optimal $C^{\ast}$ satisfies
\begin{equation}
C^{\ast}\sim(\delta\gamma)^{\frac{1}{1-\gamma}}\left(\frac{\rho}{\lambda}\right)^{\frac{1}{1-\gamma}},
\qquad\text{as $\rho\rightarrow 0$}.
\end{equation}

(ii) As $\delta$ increases, $g(\beta)$ decreases.
Since $g(\beta)$ is decreasing in $\beta$, we conclude that
$\beta$ increases as $\delta$ increases. Intuitively, it says
that if the prospect of future profits given the investment
in research and development increases, then the ruin probability decreases.
Asymptotically, as $\delta\rightarrow\infty$, we have $g(\beta)\rightarrow 0$,
and thus
$(\delta\gamma)^{\frac{1}{1-\gamma}}\left(\frac{1}{\gamma}-1\right)[g(\beta)]^{\frac{1}{1-\gamma}}
\rightarrow\rho$,
which implies that as $\delta\rightarrow\infty$, we have
$g(\beta)\sim\frac{\rho^{1-\gamma}}{\gamma\delta}\left(\frac{1}{\gamma}-1\right)^{\gamma-1}$.
Since $g(\beta)\sim\frac{1}{\beta}$ as $\beta\rightarrow\infty$, we conclude that
$\beta\sim\frac{\gamma\delta}{\rho^{1-\gamma}}\left(\frac{1}{\gamma}-1\right)^{1-\gamma}$,
as $\delta\rightarrow\infty$.
Moreover, the optimal $C^{\ast}$ satisfies:
\begin{equation}
C^{\ast}\rightarrow\frac{\rho}{\frac{1}{\gamma}-1},
\qquad
\text{as $\delta\rightarrow\infty$}.
\end{equation}

Now, if $\delta\rightarrow 0$, then $g(\beta)\rightarrow\frac{\rho}{\lambda}$.
Therefore, as $\delta\rightarrow 0$, $\beta\rightarrow\alpha$, where
we recall that $\alpha$ is the unique positive value so that
$1-\int_{0}^{\infty}e^{-\alpha y}p(y)dy=\alpha\frac{\rho}{\lambda}$,
which is the same as defined in \eqref{alphaEqn}.
Moreover, the optimal $C^{\ast}$ satisfies
\begin{equation}
C^{\ast}\sim(\delta\gamma)^{\frac{1}{1-\gamma}}\left(\frac{\rho}{\lambda}\right)^{\frac{1}{1-\gamma}},
\qquad
\text{as $\delta\rightarrow 0$}.
\end{equation}
Intuitively, it says that as $\delta\rightarrow 0$, there is no value
investing in research and development.

(iii) Similarly, as $\lambda$ increases, $\beta$ increases, and the ruin probability decreases.
As $\lambda\rightarrow\infty$, we have $g(\beta)\rightarrow 0$. Thus, $\lambda g(\beta)\rightarrow\rho$,
and $g(\beta)\sim\frac{\rho}{\lambda}$. Since $g(\beta)\sim\frac{1}{\beta}$
as $\beta\rightarrow\infty$, we conclude that
$\beta\sim\frac{\lambda}{\rho}$,
as $\lambda\rightarrow\infty$.
Moreover, the optimal $C^{\ast}$ satisfies:
\begin{equation}
C^{\ast}\sim(\delta\gamma)^{\frac{1}{1-\gamma}}\left(\frac{\rho}{\lambda}\right)^{\frac{1}{1-\gamma}},
\qquad
\text{as $\lambda\rightarrow\infty$}.
\end{equation}

(iv) Assume that the parameters are chosen so that
\begin{equation}
(\rho-\lambda\mathbb{E}[Y_{1}])-(\delta\gamma)^{\frac{1}{1-\gamma}}\left(\frac{1}{\gamma}-1\right)(\mathbb{E}[Y_{1}])^{\frac{1}{1-\gamma}}
\rightarrow 0.
\end{equation}
Then, it follows that $g(\beta)\rightarrow\mathbb{E}[Y_{1}]$ and $\beta\rightarrow 0$.
More precisely, as $\beta\rightarrow 0$, $g(\beta)\sim\mathbb{E}[Y_{1}]-\frac{\beta}{2}\mathbb{E}[Y_{1}^{2}]$
if $\mathbb{E}[Y_{1}^{2}]<\infty$, 
and \eqref{Gbeta0} becomes
\begin{equation}
\rho-(\delta\gamma)^{\frac{1}{1-\gamma}}\left(\frac{1}{\gamma}-1\right)
\left(\mathbb{E}[Y_{1}]-\frac{\beta}{2}\mathbb{E}\left[Y_{1}^{2}\right]\right)^{\frac{1}{1-\gamma}}
-\lambda\left(\mathbb{E}[Y_{1}]-\frac{\beta}{2}\mathbb{E}\left[Y_{1}^{2}\right]\right)=O(\beta^{2}),
\end{equation}
as $\beta\rightarrow 0$. Then, it follows that
\begin{align}
&\rho-(\delta\gamma)^{\frac{1}{1-\gamma}}\left(\frac{1}{\gamma}-1\right)
\left(\mathbb{E}[Y_{1}]^{\frac{1}{1-\gamma}}-\frac{1}{2(1-\gamma)}
(\mathbb{E}[Y_{1}])^{\frac{\gamma}{1-\gamma}}\mathbb{E}\left[Y_{1}^{2}\right]\beta\right)
\\
&\qquad\qquad\qquad
-\lambda\left(\mathbb{E}[Y_{1}]-\frac{\beta}{2}\mathbb{E}\left[Y_{1}^{2}\right]\right)=O(\beta^{2}),
\nonumber
\end{align}
as $\beta\rightarrow 0$.
Hence, we conclude that
\begin{equation}
\beta\sim\frac{-(\rho-\lambda\mathbb{E}[Y_{1}])+(\delta\gamma)^{\frac{1}{1-\gamma}}\left(\frac{1}{\gamma}-1\right)(\mathbb{E}[Y_{1}])^{\frac{1}{1-\gamma}}}{(\delta\gamma)^{\frac{1}{1-\gamma}}\frac{1}{2\gamma}(\mathbb{E}[Y_{1}])^{\frac{\gamma}{1-\gamma}}\mathbb{E}[Y_{1}^{2}]
+\frac{\lambda}{2}\mathbb{E}[Y_{1}^{2}]}.
\end{equation}
Moreover, the optimal $C^{\ast}$ satisfies:
\begin{equation}
C^{\ast}\sim(\delta\gamma)^{\frac{1}{1-\gamma}}(\mathbb{E}[Y_{1}])^{\frac{1}{1-\gamma}}.
\end{equation}
\end{remark}

\begin{remark}\label{ParametersII}
The value function $V(x)=e^{-\beta x}$ and
the optimal investment rate $C^{\ast}$ also depend on the parameter $\gamma$. 
We will study $\gamma=1$ case in details later. 
For the moment, let us try to understand the asymptotic behavior
of the value function and the optimal investment rate
as $\gamma\rightarrow 1^{-}$. 
We will also obtain the asymptotics as $\gamma\rightarrow 0^{+}$.
Let us recall that the optimal $C^{\ast}$ satisfies the equation:
\begin{equation}\label{Cequation}
\lambda+(1-\gamma)\delta(C^{\ast})^{\gamma}=\rho\delta\gamma(C^{\ast})^{\gamma-1}.
\end{equation}
Thus, we have $(1-\gamma)\delta(C^{\ast})^{\gamma}\leq\rho\delta\gamma(C^{\ast})^{\gamma-1}$
which implies that
$C^{\ast}\leq\frac{\rho\gamma}{1-\gamma}$.
Thus, $C^{\ast}\rightarrow 0$ as $\gamma\rightarrow 0$. 
Note that $\lim_{\gamma\rightarrow 0^{+}}\gamma^{\gamma}=1$.
Therefore, we can check that
\begin{equation}
C^{\ast}\sim\frac{\rho\delta}{\lambda+\delta}\gamma,
\qquad
\text{as $\gamma\rightarrow 0^{+}$}.
\end{equation}

Now, let us consider the $\gamma\rightarrow 1^{-}$ limit. 
Let us rewrite that equation \eqref{Cequation} as
\begin{equation}\label{Dequation}
\frac{\lambda}{(1-\gamma)^{1-\gamma}}+\delta D^{\gamma}=\frac{\rho\delta\gamma}{D^{1-\gamma}},
\end{equation}
where $D=(1-\gamma)C^{\ast}$.
Let us first consider the case $\rho\delta>\lambda$. 
Notice first that $\lim_{\gamma\rightarrow 1^{-}}(1-\gamma)^{1-\gamma}=1$.
First, $D$ cannot go to $0$ as $\gamma\rightarrow 1^{-}$, 
because otherwise the left hand side of \eqref{Dequation} goes to $\lambda$
and as $D$ goes to $0$, $D<1$ and $D^{1-\gamma}\leq 1$, so the right hand side
of \eqref{Dequation} is greater than $\rho\delta\gamma$. Then, in the limit as $\gamma\rightarrow 1^{-}$, 
we get $\lambda\geq\rho\delta$, which is a contradiction.
Second, $D$ cannot go to $\infty$ as $\gamma\rightarrow 1^{-}$. 
To see this, notice that as $D\rightarrow\infty$, the left hand side of \eqref{Dequation}
goes to $\infty$ and in the right hand side of \eqref{Dequation}, for large $D$, 
$D>1$ and $D^{1-\gamma}\geq 1$ and hence the right hand side is less than $\rho\delta$,
which is a contradiction.

Therefore, if $\rho\delta>\lambda$, $D$ converges to a positive constant,
which from \eqref{Dequation} we can see that the limit is $\frac{\rho\delta-\lambda}{\delta}$, and we have
\begin{equation}
C^{\ast}\sim\frac{\rho\delta-\lambda}{\delta}\frac{1}{1-\gamma},
\qquad
\text{as $\gamma\rightarrow 1^{-}$}.
\end{equation}
If $\rho\delta<\lambda$, then the optimal $C^{\ast}\rightarrow 0$ as $\gamma\rightarrow 1^{-}$.
To see this, notice that if $\limsup_{\gamma\rightarrow 1^{-}}C^{\ast}\in(0,\infty)$, then 
in \eqref{Cequation}, we have
$\limsup_{\gamma\rightarrow 1^{-}}\rho\delta\gamma(C^{\ast})^{\gamma-1}=\rho\delta$
and $\limsup_{\gamma\rightarrow 1^{-}}[\lambda+(1-\gamma)\delta(C^{\ast})^{\gamma}]=\lambda$,
which is a contradiction since $\rho\delta<\lambda$.
If $\limsup_{\gamma\rightarrow 1^{-}}C^{\ast}=\infty$, then
for $C^{\ast}>1$, we have from \eqref{Cequation} that 
$\lambda<\lambda+(1-\gamma)\delta(C^{\ast})^{\gamma}=\rho\delta\gamma(C^{\ast})^{\gamma-1}<\rho\delta$,
which is again a contraction. Hence, we must have $C^{\ast}\rightarrow 0$.

Since $C^{\ast}\rightarrow 0$, $(1-\gamma)\delta(C^{\ast})^{\gamma}\ll\rho\delta\gamma(C^{\ast})^{\gamma-1}$, and thus
\begin{equation}
C^{\ast}\sim\left(\frac{\lambda}{\rho\delta\gamma}\right)^{\frac{1}{\gamma-1}}
\sim
\frac{1}{e}\left(\frac{\rho\delta}{\lambda}\right)^{\frac{1}{1-\gamma}},
\qquad
\text{as $\gamma\rightarrow 1^{-}$}.
\end{equation}
If $\rho\delta=\lambda$, the optimal $C^{\ast}$ satisfies the equation:
\begin{equation}
\lambda=\frac{(1-\gamma)\delta(C^{\ast})^{\gamma}}{\gamma(C^{\ast})^{\gamma-1}-1}.
\end{equation}
Assume that $C^{\ast}>0$ is fixed, then by L'H\^{o}pital's rule,
\begin{equation}
\lim_{\gamma\rightarrow 1^{-}}
\frac{(1-\gamma)\delta(C^{\ast})^{\gamma}}{\gamma(C^{\ast})^{\gamma-1}-1}
=\lim_{\gamma\rightarrow 1^{-}}
\frac{-\delta(C^{\ast})^{\gamma}+(1-\gamma)\delta(C^{\ast})^{\gamma}\log C^{\ast}}
{(C^{\ast})^{\gamma-1}+\gamma(C^{\ast})^{\gamma-1}\log C^{\ast}}
=\frac{-\delta C^{\ast}}{1+\log C^{\ast}}.
\end{equation}
Therefore as $\gamma\rightarrow 1^{-}$, $C^{\ast}$ converges
to the unique positive solution to the equation:
$\delta x+\lambda(1+\log x)=0$.
\end{remark}


\subsubsection{The \texorpdfstring{$\gamma=1$}{gamma=1} Case}

When $\gamma=1$, it follows from Theorem~\ref{thm:indep} that
the optimal strategy $C^{\ast}$ is constant and it is given by
\begin{equation}
C^{\ast}=\argmin_{C\geq 0}\frac{\rho+C}{\lambda+\delta C}.
\end{equation}

When $\frac{\rho}{\lambda}<\frac{1}{\delta}$, then
$\inf_{C\geq 0}\frac{\rho+C}{\lambda+\delta C}=\frac{\rho}{\lambda}$
and the optimal strategy is $C_{t}\equiv 0$. 
In this case, the value function $V(x)=e^{-\beta x}$, where
\begin{equation}
\rho\beta+\lambda\int_{0}^{\infty}[e^{-\beta y}-1]p(y)dy=0.
\end{equation}

When $\frac{\rho}{\lambda}>\frac{1}{\delta}$, then
$\inf_{C\geq 0}\frac{\rho+C}{\lambda+\delta C}=\frac{1}{\delta}$. 
And for any $C\in\mathcal{C}$ and $\overline{C}:=\Vert C\Vert_{\infty}$, 
the strategy $\overline{C}$ is more optimal than $C$. The ``optimal strategy'' is $C_{t}\equiv\infty$.
Let us also assume that $\delta\mathbb{E}[Y_{1}]>1$.
In this case, the value function $V(x)=e^{-\beta x}$, where
\begin{equation}\label{betasat}
\beta+\delta\int_{0}^{\infty}[e^{-\beta y}-1]p(y)dy=0.
\end{equation}

When $\frac{\rho}{\lambda}=\frac{1}{\delta}$, in terms of ruin probability, it does not make
a difference whether the company decides to invest in research and development
or not.

\begin{remark}
When $\frac{\rho}{\lambda}\geq\frac{1}{\delta}$, $V(x)=e^{-\beta x}$, 
where $\beta$ satisfies \eqref{betasat} that is independent of $\rho$ and $\lambda$.
Asymptotically, when $\frac{\rho}{\lambda}\rightarrow 0$, it is easy to see that
$\beta\sim\frac{\lambda}{\rho}$.
\end{remark}

\begin{example}
In the special case that $p(y)=\nu e^{-\nu y}$,
when $\frac{\rho}{\lambda}<\frac{1}{\delta}$,
then the optimal $C\equiv 0$ and $V(x)=e^{-(\frac{\lambda}{\rho}-\nu)x}$,
and when $\frac{\rho}{\lambda}>\frac{1}{\delta}$
and $\frac{\delta}{\nu}>1$, then the optimal $C\equiv\infty$
and $V(x)=e^{-(\delta-\nu)x}$.
\end{example}


\section{Investing in a Market Index}\label{sec:market}

We have already studied the optimal investment in research and development
for a venture capital or high tech company in the dual risk model in Section~\ref{sec:min}, and
now, let us also add the possibility of the alternative investment
in a risky asset in the market, which is a capital market index modeled
by a geometric Brownian motion.

For simplicity, we restrict our discussions to the state-independent case as in Section~\ref{sec:indep:example}:
\begin{equation}
\rho(\cdot)\equiv\rho,
\qquad
\lambda(\cdot)\equiv\lambda,
\end{equation}
where $\rho,\lambda>0$ and
\begin{equation}
F(x,c)=\lambda+\delta c^{\gamma},
\qquad
\delta, \gamma>0.
\end{equation}

Let us assume that the market index $S_{t}$ follows a geometric Brownian motion:
\begin{equation}
dS_{t}=\mu S_{t}dt+\sigma S_{t}dW_{t},
\end{equation}
where $\mu,\sigma>0$ and $W_{t}$ is a standard Brownian motion.

Assume that at time $t$, the company can invest $\theta_{t}$ shares of the market index $S_{t}$
and $C_{t}$ in research and development.
Thus, the wealth process of the company satisfies the dynamics:
\begin{equation}\label{eqn:X^C,A}
dX_{t}=-(\rho+C_{t})dt+dJ^{C}_{t}+\theta_{t}dS_{t},
\qquad
X_{0}=x>0
\end{equation}
The invested amount in the market index is $A_{t}=\theta_{t}S_{t}$ at time $t$.

We are interested to find optimal investment strategies to minimize the probability of ruin:
\begin{equation}
V(x):=\inf_{C\in\mathcal{C},A\in\mathcal{A}}\mathbb{P}(\tau<\infty|X_{0}=x),
\end{equation}
where $\mathcal{C}$ is the same as defined before and
$\mathcal{A}$ is the admissible strategies for investment in the market index, defined as:
\begin{align}
\mathcal{A}&:=
\bigg\{A:[0,\infty)\times\Omega\rightarrow\mathbb{R}:
\text{$A$ is progressively measurable}
\\
&
\qquad\qquad\qquad\qquad\qquad\qquad
\text{and for any $t>0$, $\mathbb{E}\left[\int_{0}^{t}A_{s}^{2}ds\right]<\infty$}.\bigg\}.
\nonumber
\end{align}
For any given $C\in\mathcal{C}$ and $A\in\mathcal{A}$, we write $X^{C,A}=X$
to emphasize the dependence on $C$ and $A$.

With additional investment in a market index, 
the random time change argument in the analysis
in Section~\ref{sec:min} no longer applies. 
Instead, we rely on the stochastic optimal control theory (see e.g. Fleming--Soner \cite{Fleming-Soner-book-06}),
which suggests that the Hamilton-Jacobi-Bellman equation for $V(x)$ is given by
\begin{align}\label{HTwo}
&\inf_{C\geq 0,A\in\mathbb{R}}
\bigg\{-(\rho+C)V'(x)+(\lambda+\delta C^{\gamma})\int_{0}^{\infty}[V(x+y)-V(x)]p(y)dy
\\
&\qquad\qquad\qquad
+A\mu V'(x)+\frac{1}{2}A^{2}\sigma^{2}V''(x)\bigg\}=0,
\nonumber
\end{align}
with boundary condition $V(0)=1$.

Similar as in Section~\ref{sec:min}, the case $\gamma\geq 1$ leads to triviality
and for the rest, we consider two cases: $0<\gamma<1$ and $\gamma=1$.

\subsection{The \texorpdfstring{$0<\gamma<1$}{0<gamma<1} Case}

In this section, we consider the $0<\gamma<1$ case.
We start with the following technical lemma.

\begin{lemma}\label{tech:lemma:smaller:than:1}
$V(x)=e^{-\beta x}$ is a solution to the Hamilton-Jacobi-Bellman equation \eqref{HTwo},
where $\beta>0$ is the unique solution to the equation:
\begin{align}\label{betaEquation}
&\beta\left[\rho+\left(\frac{1}{\delta\gamma}\right)^{\frac{1}{\gamma-1}}
\left(\frac{\beta}{1-\int_{0}^{\infty}e^{-\beta y}p(y)dy}\right)^{\frac{1}{\gamma-1}}\right]
\\
&\qquad
-\left[\lambda+\delta\left(\frac{1}{\delta\gamma}\right)^{\frac{\gamma}{\gamma-1}}
\left(\frac{\beta}{1-\int_{0}^{\infty}e^{-\beta y}p(y)dy}\right)^{\frac{\gamma}{\gamma-1}}\right]
\left(1-\int_{0}^{\infty}e^{-\beta y}p(y)dy\right)
-\frac{1}{2}\frac{\mu^{2}}{\sigma^{2}}=0.
\nonumber
\end{align}
Given $V(x)=e^{-\beta x}$
and let
\begin{align}
(C^{\ast},A^{\ast})&\in{\rm argmin}\bigg\{-(\rho+C)V'(x)+(\lambda+\delta C^{\gamma})\int_{0}^{\infty}[V(x+y)-V(x)]p(y)dy\nonumber
\\
&\qquad\qquad\qquad\qquad\qquad
+A\mu V'(x)+\frac{1}{2}A^{2}\sigma^{2}V''(x)\bigg\}.
\label{defn:C^*A^*}
\end{align}
Then, we have
\begin{equation}\label{C:star:A:star}
C^{\ast}=\left(\frac{1}{\delta\gamma}\right)^{\frac{1}{\gamma-1}}
\left(\frac{\beta}{1-\int_{0}^{\infty}e^{-\beta y}p(y)dy}\right)^{\frac{1}{\gamma-1}},
\qquad
A^{\ast}=\frac{\mu}{\sigma^{2}\beta}.
\end{equation}
\end{lemma}

\begin{proof}[Proof of Lemma~\ref{tech:lemma:smaller:than:1}]
Assume that $V'(x)<0$ and $V''(x)>0$, then, the optimal $C$ and $A$
are given respectively by
\begin{align}
C=\left(\frac{1}{\delta\gamma}\right)^{\frac{1}{\gamma-1}}
\left(\frac{V'(x)}{\int_{0}^{\infty}[V(x+y)-V(x)]p(y)dy}\right)^{\frac{1}{\gamma-1}},
\qquad
A=-\frac{\mu V'(x)}{\sigma^{2}V''(x)},
\end{align}
and the Hamilton-Jacobi-Bellman equation becomes
\begin{align}
&-\left[\rho+\left(\frac{1}{\delta\gamma}\right)^{\frac{1}{\gamma-1}}
\left(\frac{V'(x)}{\int_{0}^{\infty}[V(x+y)-V(x)]p(y)dy}\right)^{\frac{1}{\gamma-1}}\right]
V'(x)
\\
&\qquad
+\left[\lambda+\delta\left(\frac{1}{\delta\gamma}\right)^{\frac{\gamma}{\gamma-1}}
\left(\frac{V'(x)}{\int_{0}^{\infty}[V(x+y)-V(x)]p(y)dy}\right)^{\frac{\gamma}{\gamma-1}}\right]
\nonumber
\\
&\qquad\qquad\qquad\qquad\cdot
\int_{0}^{\infty}[V(x+y)-V(x)]p(y)dy
-\frac{1}{2}\frac{\mu^{2}}{\sigma^{2}}\frac{(V'(x))^{2}}{V''(x)}=0.
\nonumber
\end{align}
We can see that $V(x)=e^{-\beta x}$, where $\beta>0$ is the unique solution to the equation:
\begin{align}
&\beta\left[\rho+\left(\frac{1}{\delta\gamma}\right)^{\frac{1}{\gamma-1}}
\left(\frac{\beta}{1-\int_{0}^{\infty}e^{-\beta y}p(y)dy}\right)^{\frac{1}{\gamma-1}}\right]
\\
&\qquad
-\left[\lambda+\delta\left(\frac{1}{\delta\gamma}\right)^{\frac{\gamma}{\gamma-1}}
\left(\frac{\beta}{1-\int_{0}^{\infty}e^{-\beta y}p(y)dy}\right)^{\frac{\gamma}{\gamma-1}}\right]
\left(1-\int_{0}^{\infty}e^{-\beta y}p(y)dy\right)
-\frac{1}{2}\frac{\mu^{2}}{\sigma^{2}}=0.
\nonumber
\end{align}
Recall the definition $g(\beta)=\frac{1}{\beta}\left[1-\int_{0}^{\infty}e^{-\beta y}p(y)dy\right]$
and we want to show that the equation
\begin{equation}
H(\beta):=\rho-(\delta\gamma)^{\frac{1}{1-\gamma}}\left(\frac{1}{\gamma-1}\right)[g(\beta)]^{\frac{1}{1-\gamma}}
-\lambda g(\beta)-\frac{1}{2}\frac{\mu^{2}}{\sigma^{2}}\frac{1}{\beta}=0
\end{equation}
has a unique positive solution.
It is easy to see that $\lim_{\beta\rightarrow 0^{+}}g(\beta)=\mathbb{E}[Y_{1}]$ and 
$\lim_{\beta\rightarrow\infty}g(\beta)=0$. 
Thus, $H(\beta)\sim-\frac{1}{2}\frac{\mu^{2}}{\sigma^{2}\beta}<0$
as $\beta\rightarrow 0^{+}$
and $H(\beta)\rightarrow\rho$ as $\beta\rightarrow\infty$.
We have already proved that $g(\beta)$ is decreasing in $\beta$.
Moreover, $\frac{1}{\beta}$ is also decreasing in $\beta$. 
Therefore $H(\beta)$ is increasing in $\beta$ and hence
there exists a unique positive value $\beta$ so that $H(\beta)=0$.

Finally, we can compute that the optimal $C^{\ast}$ and $A^{\ast}$
are given by \eqref{C:star:A:star}.
This completes the proof.
\end{proof}
\subsubsection{A Verification Theorem}

Let us recall from \eqref{HTwo} that the Hamilton-Jacobi-Bellman equation is given by
\begin{equation}\label{eqn:(5.2)}
\begin{split}
0&=\mathop{\inf}\limits_{C>0,A\in\mathbb{R}}\bigg\{-(\rho+C)V'(x)+(\lambda+\delta C^{\gamma})\int_{0}^{\infty}[V(x+y)-V(x)]p(y)dy\\
&\hspace{7cm}+A\mu V'(x)+\frac12A^2\sigma^2 V''(x)
\bigg\},
\end{split}
\end{equation}
with boundary condition $V(0)=1$. 

\begin{theorem}[Vertification]\label{eqn:investment_verification_gamma<1}
If $w\in{\rm C}^2_{\rm b}$ is a solution of \eqref{eqn:(5.2)} with  $w(0)=1$, such that for any $C\in\mathcal{C}$ and $A\in\mathcal{A}$
\begin{equation}\label{eqn:w(X)1{tau<infty}=0}
\lim_{K\rightarrow\infty}w(K)=0,
\end{equation}
then, $w\le V$. In addition, if 
\[
C^*(x):=\left(\frac{1}{\delta\gamma}\right)^{\frac{1}{\gamma-1}}
\left(\frac{w'(x)}{\int_{0}^{\infty}[w(x+y)-w(x)]p(y)dy}\right)^{\frac{1}{\gamma-1}}\;\text{ \rm and }\;A^*(x)=-\frac{\mu w'(x)}{\sigma^{2}w''(x)},
\]
are such that 
\[
dX^*_t=-(\rho+C^*(X^*_t))dt+dJ_t^{C^*(X^*_{t-})}+A^*(X^*_t)dS_t
\]
has a solution 
and $C^*_\cdot:=C^*(X^*_\cdot)\in\mathcal{C}$ and  $A^*_\cdot:=A^*(X^*_\cdot)\in\mathcal{A}$, then $w=V$.
\end{theorem}

\begin{proof}[Proof of Theorem~\ref{eqn:investment_verification_gamma<1}]
We follow the supermartingale argument presented in \cite[Theorem~1.1]{Rogers}.
Since $w$ is bounded and continuously differentiable with bounded derivative, by It\^o lemma for jump processes we have
\begin{align}
\mathbb{E}\left[w\left(X^{C,A}_{t}\right)\Big|\mathcal{F}_{s}\right]&=w\left(X_{s}^{C,A}\right)+\mathbb{E}\bigg[\int_s^{t}\bigg(-(\rho+C_u)w'\left(X^{C,A}_u\right)
\\
&\qquad\qquad
+(\lambda+\delta C_u^\gamma)\int_0^\infty\left[w\left(X^{C,A}_u+y\right)-w\left(X^{C,A}_u\right)\right]p(y)dy
\nonumber
\\
&\qquad\qquad
+A_{u}\mu w'\left(X^{C,A}_u\right)+\frac{1}{2}A_{u}^{2}\sigma^{2}w''\left(X^{C,A}_u\right)\bigg)du\bigg|\mathcal{F}_s\bigg]\ge w\left(X_{s}^{C,A}\right),
\nonumber
\end{align}
for any $C\in\mathcal{C}$ and $A\in\mathcal{A}$. Therefore, 
$w(X_{t}^{C,A})$ is a submartingale. 
Let $\tau_{K}$ be the first time that the $X_{t}^{C,A}$ process hits $K>0$.
Since $w$ is uniformly bounded, by optional stopping theorem, 
\[
w(x)\le \mathbb{E}\left[w\left(X^{C,A}_{\tau_{K}\wedge \tau}\right)\right]
=\mathbb{E}\left[w\left(X^{C,A}_{\tau_{K}}\right)1_{\{\tau_{K}<\tau\}}+1_{\{\tau_{K}\ge\tau\}}\right]
=w(K)\mathbb{P}(\tau_{K}<\tau)+\mathbb{P}(\tau<\tau_{K}).
\]
It follows from \eqref{eqn:w(X)1{tau<infty}=0} and 
monotone convergence theorem that the right hand side above converges to $\mathbb{P}(\tau<\infty)$ 
as $K\rightarrow\infty$ and thus
\[
w(x)\le \mathbb{P}(\tau<\infty).
\]
By taking infimum over $C\in\mathcal{C}$ and $A\in\mathcal{A}$, we obtain $w\le V$. All the above inequalities change to equality for $C_{t}=C^{\ast}(X_{t-}^{\ast})$ and $A_{t}=A^{\ast}(X_{t-}^{\ast})$. This completes the proof.
%
\end{proof}

\begin{corollary}\label{cor:w=V_gamma<1}
$w(x)=e^{-\beta x}$ with $\beta$ defined in \eqref{betaEquation} satisfies \eqref{eqn:w(X)1{tau<infty}=0} and thus $w=V$.
\end{corollary}

\begin{proof}[Proof of Corollary~\ref{cor:w=V_gamma<1}]
We already showed, in Lemma~\ref{tech:lemma:smaller:than:1}, that $w$ is a classical solution of the boundary value problem \eqref{eqn:(5.2)}. 
Moreover, since  $C^*$ and $A^*$ defined by \eqref{defn:C^*A^*}
are admissible controls (constants). By Theorem~\ref{eqn:investment_verification_gamma<1} and because \eqref{eqn:w(X)1{tau<infty}=0} trivially holds, we have $V(x)=w(x)=e^{-\beta x}$.
The proof is complete.
\end{proof}

Next, we provide some asymptotic analysis.

\begin{remark}
As in Remark \ref{ParametersI}, let us
discuss the dependence of $C^{\ast}$, $\beta$ and hence $V(x)=e^{-\beta x}$
on the parameters $\rho$, $\lambda$ and $\delta$.
Since the results are similar to Remark \ref{ParametersI}, we omit
the details and only summarize the results here.
Note that $\beta$ satisfies
\begin{equation}
\rho-(\delta\gamma)^{\frac{1}{1-\gamma}}\left(\frac{1}{\gamma}-1\right)[g(\beta)]^{\frac{1}{1-\gamma}}
-\lambda g(\beta)-\frac{1}{2}\frac{\mu^{2}}{\sigma^{2}}\frac{1}{\beta}=0,
\end{equation}
where $g(\beta)$ is defined in \eqref{small:g:eqn}.

(i) As $\rho\rightarrow 0^{+}$, we have
$\beta\sim\frac{\lambda+\frac{1}{2}\frac{\mu^{2}}{\sigma^{2}}}{\rho}$,
and
$C^{\ast}\sim(\delta\gamma)^{\frac{1}{1-\gamma}}\left(\frac{\rho}{\lambda+\frac{1}{2}\frac{\mu^{2}}{\sigma^{2}}}\right)^{\frac{1}{1-\gamma}}$.

(ii) As $\delta\rightarrow\infty$, we have
$\beta\sim\frac{\gamma}{\rho^{1-\gamma}}\left(\frac{1}{\gamma}-1\right)^{1-\gamma}\delta$,
and
$C^{\ast}\rightarrow\frac{\rho}{\frac{1}{\gamma}-1}$.
As $\delta\rightarrow 0$, we have $\beta\rightarrow\alpha$, where
$\alpha$ is the unique positive value so that
\begin{equation}
\rho\alpha+\lambda\int_{0}^{\infty}[e^{-\alpha y}-1]p(y)dy-\frac{1}{2}\frac{\mu^{2}}{\sigma^{2}}=0.
\end{equation}
Moreover, as $\delta\rightarrow 0$, we have
$C^{\ast}\sim(\delta\gamma)^{\frac{1}{1-\gamma}}\left(\frac{1}{\lambda}\left(\rho-\frac{1}{2\alpha}\frac{\mu^{2}}{\sigma^{2}}\right)\right)^{\frac{1}{1-\gamma}}$.

(iii) As $\lambda\rightarrow\infty$, we have
$\beta\sim\frac{\lambda}{\rho}$,
and
$C^{\ast}\sim(\delta\gamma)^{\frac{1}{1-\gamma}}\left(\frac{\rho}{\lambda}\right)^{\frac{1}{1-\gamma}}$.
\end{remark}

\begin{remark}
Here, we investigate the asymptotic behavior
of the value function and the optimal investment rate
as $\gamma\rightarrow 1^{-}$ and $\gamma\rightarrow 0^{+}$. 
Note that the optimal $C^{\ast}$ and $\beta$ satisfy:
\begin{equation}
\rho-\left(\frac{1}{\gamma}-1\right)C^{\ast}-\frac{\lambda}{\delta\gamma}(C^{\ast})^{1-\gamma}-\frac{1}{2}\frac{\mu^{2}}{\sigma^{2}}\frac{1}{\beta}=0,
\end{equation}
and
\begin{equation}
C^{\ast}=\left(\frac{1}{\delta\gamma}\right)^{\frac{1}{\gamma-1}}
\left(\frac{\beta}{1-\int_{0}^{\infty}e^{-\beta y}p(y)dy}\right)^{\frac{1}{\gamma-1}}.
\end{equation}

(i) As $\gamma\rightarrow 0^{+}$, $C^{\ast}\sim\eta\gamma$ for some $\eta>0$
and $\beta\rightarrow\iota$ for some $\iota>0$.
It is easy to check that $\eta,\iota>0$ satisfy:
$\eta=\frac{1-\int_{0}^{\infty}e^{-\iota y}p(y)dy}{\iota}$
and $\rho-\eta-\frac{\lambda}{\delta}\eta-\frac{1}{2}\frac{\mu^{2}}{\sigma^{2}}\frac{1}{\iota}=0$. Thus
\begin{equation}
\rho-\left(1+\frac{\lambda}{\delta}\right)\frac{1-\int_{0}^{\infty}e^{-\iota y}p(y)dy}{\iota}
-\frac{1}{2}\frac{\mu^{2}}{\sigma^{2}}\frac{1}{\iota}=0.
\end{equation}

(ii) Next, let us consider $\gamma\rightarrow 1^{-}$.

If $\delta\mathbb{E}[Y_{1}]>1$, then there exists a unique
value $\iota>0$ such that 
$\delta=\frac{\iota}{1-\int_{0}^{\infty}e^{-\iota y}p(y)dy}$.
Assume further that 
$\rho-\frac{\lambda}{\delta}-\frac{1}{2}\frac{\mu^{2}}{\sigma^{2}\iota}>0$.
Then, we have $C^{\ast}\sim\frac{\eta}{1-\gamma}$ and $\beta\rightarrow\iota$ as $\gamma\rightarrow 1^{-}$,
where 
$\eta=\rho-\frac{\lambda}{\delta}-\frac{1}{2}\frac{\mu^{2}}{\sigma^{2}\iota}$.

If $\rho-\frac{\lambda}{\delta}-\frac{1}{2}\frac{\mu^{2}}{\sigma^{2}\iota}<0$, 
the optimal $C^{\ast}\rightarrow 0$ as $\gamma\rightarrow 1^{-}$
and $C^{\ast}\sim\left(\frac{\delta\gamma}{\lambda}\left(\rho-\frac{1}{2}\frac{\mu^{2}}{\sigma^{2}}\frac{1}{\beta}\right)\right)^{\frac{1}{1-\gamma}}$
and $\beta\rightarrow\iota$
as $\gamma\rightarrow 1^{-}$.
We can check that $\eta,\iota$ satisfy the equations:
$\eta=\frac{\lambda}{\delta\left(\rho-\frac{1}{2}\frac{\mu^{2}}{\sigma^{2}\iota}\right)}$
and
$\frac{\iota}{\delta}=1-\int_{0}^{\infty}e^{-\iota y}p(y)dy$.
As $\gamma\rightarrow 1^{-}$, we have
$C^{\ast}\sim\frac{1}{e}\left(\frac{\delta}{\lambda}\left(\rho-\frac{1}{2}\frac{\mu^{2}}{\sigma^{2}\iota}\right)\right)^{\frac{1}{1-\gamma}}$.

If $\rho-\frac{\lambda}{\delta}-\frac{1}{2}\frac{\mu^{2}}{\sigma^{2}\iota}=0$, 
then, as $\gamma\rightarrow 1^{-}$, we have that $C^{\ast}$ converges
to the unique positive solution to the equation:
$\delta x+\lambda(1+\log x)=0$.
\end{remark}


\subsection{The \texorpdfstring{$\gamma=1$}{gamma=1} Case}

Consider the case where $\gamma=1$, i.e. for $x>0$. 
Then we have a singular control problem on $C\in\mathcal{C}$ (see e.g. Fleming--Soner \cite{Fleming-Soner-book-06}) 
and the value function $V(x)$
satisfies the Hamilton-Jacobi-Bellman equation:
\begin{equation}
\begin{split}
0&=\min\Bigg\{-\rho V'(x)+\lambda\int_{0}^{\infty}[V(x+y)-V(x)]p(y)dy
+\inf_{A\in\mathbb{R}}\left\{A\mu V'(x)+\frac{1}{2}A^{2}\sigma^{2}V''(x)\right\},\\
&\hspace{5cm}\delta\int_{0}^{\infty}[V(x+y)-V(x)]p(y)dy-V'(x)\Bigg\}\,,
\end{split}
\end{equation}
with boundary condition $V(0)=1$.
Optimizing over $A$, it reduces to the following equation:
\begin{equation}\label{eqn:(5.12)}
\begin{split}
0&=\min\Bigg\{-\rho V'(x)+\lambda\int_{0}^{\infty}[V(x+y)-V(x)]p(y)dy-\frac{\mu^2(V')^2}{2\sigma^2V''},\\
&\hspace{5cm}\delta\int_{0}^{\infty}[V(x+y)-V(x)]p(y)dy-V'(x)\Bigg\}\,,
\end{split}
\end{equation}
with boundary condition $V(0)=1$.

For $w\in{\rm C}^2_{\rm b}$, we define
\[
\mathcal{P}:=\bigg\{x\in\mathbb{R}_+\;:\; \delta\int_{0}^{\infty}[w(x+y)-w(x)]p(y)dy-w'(x)>0\bigg\}.
\]
According to Fleming--Soner \cite[Chapter 8]{Fleming-Soner-book-06}, $w$ is a classical solution of \eqref{eqn:(5.12)}  if 

(i) On $\mathcal{P}$, $w$ satisfies
\[
0=-\rho w'(x)+\lambda\int_{0}^{\infty}[w(x+y)-w(x)]p(y)dy-\frac{\mu^2(w')^2}{2\sigma^2w''}.
\] 

(ii)  On $\mathbb{R}_+$, $w$ satisfies
\begin{equation}\label{eqn:variational_inequality}
\begin{split}
0&\le-\rho w'(x)+\lambda\int_{0}^{\infty}[w(x+y)-w(x)]p(y)dy-\frac{\mu^2(w')^2}{2\sigma^2w''},\\
0&\le \delta\int_{0}^{\infty}[w(x+y)-w(x)]p(y)dy-w'(x).
\end{split}
\end{equation}

(iii) $w(0)=1$.

\begin{lemma}\label{lem:classical_solution_investment}
$w(x)=e^{-(\beta_1\vee\beta_2)x}$ is a classical solution of \eqref{eqn:(5.12)} where $\beta_1$ is the unique positive solutions of $F(\beta)=0$ and $\beta_2$ is the unique positive solution of $G(\beta)=0$ if it exists or zero otherwise. Here $F$ and $G$ are given by
\[
\begin{split}
F(\beta)&:=\rho\beta+\lambda\int_0^\infty[e^{-\beta y}-1]p(y)dy-\frac12\frac{\mu^2}{\sigma^2},\\
G(\beta)&:=\beta+\delta\int_0^\infty[e^{-\beta y}-1]p(y)dy.
 \end{split}
\]
\end{lemma}

\begin{proof}[Proof of Lemma~\ref{lem:classical_solution_investment}]
If $G'(0)=1-\delta\mathbb{E}[Y_1]\ge0$, then $\beta_2=0$ and $G(\beta_1)>0$. This implies that $\mathcal{P}=\mathbb{R}_+$. 
By straightforward calculations, 
\[
\begin{split}
-\rho w'(x)+\lambda\int_{0}^{\infty}[w(x+y)-w(x)]p(y)dy-\frac{\mu^2(w')^2}{2\sigma^2w''}=wF(\beta_1)&=0,\\
\delta\int_{0}^{\infty}[w(x+y)-w(x)]p(y)dy-w'(x)=wG(\beta_1)&>0.
\end{split}
\]
If $G'(0)=1-\delta\mathbb{E}[Y_1]<0$ and $\beta_1> \beta_2$, then $G(\beta_1)>0$ and we have $\mathcal{P}=\mathbb{R}_+$. Similar to the previous paragraph we obtain that $w$ is a classical solution.
If $G'(0)=1-\delta\mathbb{E}[Y_1]<0$ and $\beta_1\le \beta_2$, then $F(\beta_2)\ge0$ and we have $\mathcal{P}=\emptyset$. Thus,
\[
\begin{split}
-\rho w'(x)+\lambda\int_{0}^{\infty}[w(x+y)-w(x)]p(y)dy-\frac{\mu^2(w')^2}{2\sigma^2w''}=wF(\beta_2)&\ge0,\\
\delta\int_{0}^{\infty}[w(x+y)-w(x)]p(y)dy-w'(x)=wG(\beta_2)&=0.
\end{split}
\]
The proof is complete.
\end{proof}

\subsubsection{A Verification Theorem}

\begin{theorem}[Verification]\label{thm:verification_investment_gamma=1}
Let $w\in{\rm C}^2_{\rm b}$ be a decreasing classical solution of problem \eqref{eqn:(5.12)} such that condition \eqref{eqn:w(X)1{tau<infty}=0} holds.
Then, $w(x)\le V(x)$, where $V(x)$ is the value function of the ruin probability minimization problem with investment. \\
In addition, if $\mathcal{P}=\mathbb R_+$,  then $w(x)=V(x)$.
\end{theorem}

\begin{proof}[Proof of Theorem~\ref{thm:verification_investment_gamma=1}]
Let $A=\{A_s\}_{s\ge0}$ be an admissible strategy and $C:=\{C_t\}_{t\ge0}$ be  a non-decreasing singular function, i.e. $C_t:=\int_0^tdc_s$ where $c_s$ is a non-negative measure. Then, 
\[
X^{C,A}_t=x-\rho t-C_t+J_t^C+\int_0^tA_sdS_s\,,
\]
where $J_t^C=\sum_{i=1}^{N_t^C}Y_i$ where $N_t^C$ is a simple point process with compensator  $\lambda t+\delta C_t$.
Then, by It\^o's formula for ${\rm C}^2_{\rm b}$ functions,  we have 
\[
\begin{split}
\mathbb{E}\left[w\left(X^{C,A}_{t}\right)\Big|\mathcal{F}_{s}\right]&=w\left(X^{C,A}_{s}\right)+\mathbb{E}\Bigg[\int_s^{t}\bigg(-\rho w'+\lambda\int_{0}^{\infty}[w(x+y)-w(x)]p(y)dy\\
&+A_u\mu w'+\frac12A_u^2\sigma^2 w''\bigg)\left(X^{C,A}_{u}\right)du\\
&+\int_s^{t}\bigg(-w'+\delta\int_{0}^{\infty}[w(x+y)-w(x)]p(y)dy\bigg)\left(X^{C,A}_{u}\right)dC^0_u\\
&+\sum_{s\le u\leq t}\left(w(X^{C,A}_{u}-\Delta C_u)-w\left(X^{C,A}_{u}\right)\right)\Bigg].
\end{split}
\]
Here $C_u=C^0_u+\Delta C_u$ where $C^0_u$ is the continuous part of $C$ and $\Delta C_u$ is the pure jump part of $C_u$. Notice that by the definition of classical solution, \eqref{eqn:variational_inequality} holds and therefore, the first two terms inside the expectation above are non-negative. In addition since $w$ is non-increasing, we have $w(X^{C,A}_{u}-\Delta C_u)-w(X^{C,A}_{u})\ge0$. Thus, $\mathbb{E}[w(X^{C,A}_{t})|\mathcal{F}_{s}]\ge w(X^{C,A}_{s})$ and $w(X^{C,A}_{t})$ is a submartingale.
Similar to the arguments in the proof of Theorem \ref{eqn:investment_verification_gamma<1}, \eqref{eqn:w(X)1{tau<infty}=0} implies that $w(x)\le \mathbb{P}(\tau<\infty)$. By taking the infimum over $(C,A)$, we obtain $w\le V$.\\
Now assume that $\mathcal{P}=\mathbb{R}_+$ and set $C\equiv0$. It follows from the definition of $A^*$ and It\^o's formula that
 \[
\begin{split}
\mathbb{E}\left[w\left(X^*_{t\wedge \tau}\right)\right]&=w(x)+\mathbb{E}\Bigg[\int_0^{t\wedge \tau}\bigg(-\rho w'+\lambda\int_{0}^{\infty}[w(x+y)-w(x)]p(y)dy\\
&+A^*\mu w'+\frac12(A^*)^2\sigma^2 w''\bigg)(X^*_{s})ds\Bigg]=w(x),
\end{split}
\]
In the above, $X^*$ satisfies $X^*_t=x-\rho t+J_t^\lambda+\int_0^tA^*(X^*_s)dW_s$.
If we let $t\to\infty$, we obtain $w(x)=\mathbb{P}(\tau^*<\infty)\ge V(x)$ where $\tau^*$ is the ruin time for process $X^*$.
The proof is complete
\end{proof}

\begin{corollary}\label{cor:w=V_gamma=1}
The classical solution $w(x)=e^{-(\beta_1\vee\beta_2)x}$ of boundary value problem \eqref{eqn:(5.2)} satisfies the assumption of the verification and thus $w=V$.
\end{corollary}

\begin{proof}[Proof of Corollary~\ref{cor:w=V_gamma=1}]
First, the condition \eqref{eqn:w(X)1{tau<infty}=0} trivially holds. Therefore,
if $\beta_1>\beta_2$, then $\mathcal{P}=\mathbb{R}_+$ and $w=V$ is followed by Theorem \ref{thm:verification_investment_gamma=1}. It remains to show the result for the case that when $\beta_1\le\beta_2$, i.e. $\mathcal{P}=\emptyset$. 
For $c>0$ let $w_c(x)=\mathbb{P}(\tau_c<\infty)$ with $X_t=x-(\rho+c) t+J_t^c+\int_0^tA^*dW_s$ with $A^*=\frac{\mu}{\sigma^2\beta_2}$. Then, immediately we obtain $w_c\ge V$. We want to show that $w_{c}(x)\to w(x)=e^{-\beta_2 x}$ as $c\to\infty$. 
Notice that $w_c$ satisfies the equation
\[
0=-(\rho+c) w_c'(x)+(\lambda+\delta c)\int_{0}^{\infty}[w_c(x+y)-w_c(x)]p(y)dy-\frac{\mu^2(w_c')^2}{2\sigma^2w_c''},
\]
with the boundary condition $w_c(0)=1$.
The unique bounded solution of the above equation is given by $w_c(x)=e^{-\beta(c)x}$ where $\beta(c)$ satisfies 
\begin{equation}\label{eqn:beta(c)}
-(\rho+c)\beta(c)+(\lambda+\delta c)\int_0^\infty[e^{-\beta(c)y}-1]p(y)dy-\frac{\mu^2}{2\sigma^2}=0.
\end{equation}
Notice that for any $c>0$, $\beta(c)$ is uniquely determined and is continuous on $c$. In addition, straightforward calculations shows that $\beta(c)$ is increasing, i.e.
\[
\beta'(c)=\frac{1}{c}\frac{\rho+\lambda\int_0^\infty[1-e^{-\beta(c)y}]p(y)dy+\frac{\mu^2}{2\sigma^2}}{\rho+c+(\lambda+\delta c)\int_0^\infty e^{-\beta(c)y}yp(y)dy}>0.
\]
Thus, $\bar\beta:=\lim_{c\to\infty}\beta(c)$ exists and $\bar{\beta}>0$ 
and after dividing \eqref{eqn:beta(c)} by $c$ and taking limit when $c\to\infty$, we obtain
\[
G(\bar\beta)=-\bar\beta+\lambda\int_0^\infty\left[e^{-\bar\beta y}-1\right]p(y)dy=0.
\]
Since $G$ has a unique positive solution, we must have $\bar\beta=\beta_2$ 
and therefore, we obtain $V(x)\le \lim_{c\to\infty}w_c(x)=e^{-\beta_2 x}$.
This completes the proof.
\end{proof}


\section{Numerical Studies}\label{sec:numerical}

In this section, we carry out numerical studies to illustrate and understand better
how the minimized ruin probability and the optimal investment rate depend
on the parameters in the dual risk model.

\subsection{State-Independent Ruin Probability with Optimal Investment}

In this section, we assume that the dual risk model is state-independent, 
and in particular, we assume that
$\rho(\cdot)\equiv\rho$,
$\lambda(\cdot)\equiv\lambda$,
and
$F(\cdot,c)\equiv 
\lambda+\delta c^{\gamma}$.
We also assume that $Y_{i}$ are i.i.d. exponentially distributed so that $p(y)=\nu e^{-\nu y}$ for some $\nu>0$.
We also assume that $\lambda\mathbb{E}[Y_{1}]=\frac{\lambda}{\nu}>\rho$ so that
the ruin probability is less than $1$ without any investment in research and development. 
Indeed, the ruin probability is given by $e^{-\alpha x}$, where, according to \eqref{alphaEqn}, $\alpha$ satisfies the equation:
\begin{equation}
\rho\alpha+\lambda\int_{0}^{\infty}[e^{-\alpha y}-1]\nu e^{-\nu y}dy
=\rho\alpha-\lambda\frac{\alpha}{\nu+\alpha}=0,
\end{equation}
which implies that $\alpha=\frac{\lambda}{\rho}-\nu$.

In Figure~\ref{Comparison}, we compare 
the ruin probability without any investment, the minimized ruin probability
with investment in research and development, and the minimized ruin probability when investment in 
both research and development and a market index are allowed.
For simplicity, we assume that $\gamma=\frac{1}{2}$ so that 
as in Example \ref{ExpExample}, the minimized ruin probability is $V(x)=e^{-\beta x}$, where
$\beta=\frac{\lambda+\sqrt{\lambda^{2}+\rho\delta^{2}}}{2\rho}-\nu$,
and by investing in research and development, it reduces the ruin probability.
Now, if additional investment in a risky asset, e.g. a market index is allowed, then the ruin probability
can be further reduced and the minimized ruin probability becomes $V(x)=e^{-\beta x}$, 
where by letting $p(y)=\nu e^{-\nu y}$ and $\gamma=\frac{1}{2}$
in \eqref{betaEquation}, we deduce that $\beta>0$ is the unique solution to the equation:
\begin{equation}
\beta\rho-\frac{\beta\delta^{2}}{4}\frac{1}{(\nu+\beta)^{2}}
-\frac{\lambda\beta}{\nu+\beta}-\frac{1}{2}\frac{\mu^{2}}{\sigma^{2}}=0.
\end{equation}

\begin{figure}[htb]
\begin{center}
\includegraphics[scale=1.0]{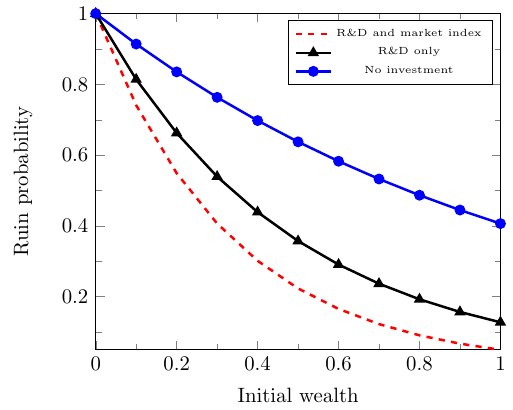}
\caption{Illustration of the ruin probability without any investment (blue curve with circle markers), the minimized ruin probability
with investment in research and development (black curve with triangle markers), and the minimized ruin probability when investment in 
both research and development and a market index are allowed (red dashed curve). The $x$-axis denotes
the initial wealth of the underlying company and the $y$-axis denotes the (minimized) ruin probability.
Here, we take $\gamma=\frac{1}{2}$, $\rho=0.1$, $\nu=0.1$, 
$\lambda=0.1$, $\delta=1$, $\mu=0.1$ and $\sigma=0.2$.}
\label{Comparison}
\end{center}
\end{figure}

In Figure~\ref{GammaDelta}, we investigate the dependence of the optimal $C^{\ast}$ on
the parameters $\gamma$ and $\delta$ given $\rho=2$, $\nu=2$, and $\lambda=0.1$.
Let us recall that when investment in research and development is allowed, the optimal investment rate
$C^{\ast}$ is the unique positive solution to the following equation:
\begin{equation}\label{CstarEqn}
\lambda+(1-\gamma)\delta(C^{\ast})^{\gamma}=\rho\delta\gamma(C^{\ast})^{\gamma-1}.
\end{equation}
When additional investment in a market index is allowed, the optimal investment rate $C^{\ast}$
for the investment in research and development remains the same. 
Notice that from \eqref{CstarEqn}, the optimal $C^{\ast}$ is independent of the distribution
of $Y_{i}$. And therefore the definition of $C^{\ast}$ is independent of the condition
\eqref{ConditionOne} under which the minimized ruin probability is less than $1$. 
Intuitively, that is because, $C^{\ast}$ optimizes over the drift term by the random time change technique, 
but when the condition \eqref{ConditionOne} is violated, even the optimal $C^{\ast}$ still gives
the ruin probability equal to $1$. In Figure~\ref{GammaDelta}, we give
the heat map plot of the optimal $C^{\ast}$ as function of $\gamma$ and $\delta$. 
Note that for $p(y)=\nu e^{-\nu y}$
the condition \eqref{ConditionOne} is equivalent to
\begin{equation}\label{boundary}
\rho-\frac{\lambda}{\nu}-(\delta\gamma)^{\frac{1}{1-\gamma}}\left(\frac{1}{\gamma}-1\right)\frac{1}{\nu^{\frac{1}{1-\gamma}}}<0.
\end{equation}
When this condition is violated, 
then it corresponds to the darker region in the bottom half of the plot in Figure~\ref{GammaDelta}. 
The boundary is achieved when the left hand side of \eqref{boundary} is zero.
In this region, the ruin probability is always $1$ regardless of the investment in research and development.
When the condition \eqref{boundary} is satisfied,
it corresponds to the upper half of the plot in Figure~\ref{GammaDelta}. 
In this region, it is easy to observe that as $\delta$ increases, $C^{\ast}$ increases.
For the plot in Figure~\ref{GammaDelta}, the optimal $C^{\ast}$ is less sensitive
to the change of the parameter $\gamma$.

\begin{figure}[htb]
\begin{center}
\includegraphics[scale=1.0]{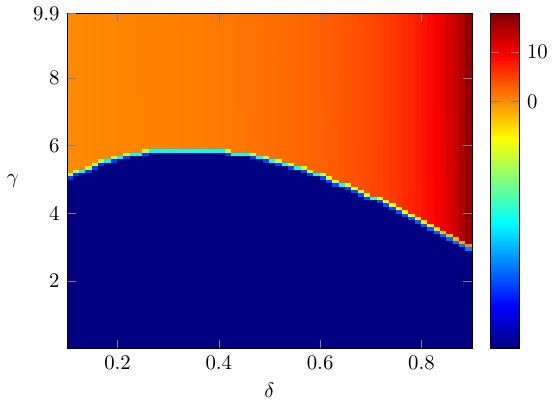}
\caption{This shows $C^{\ast}$ as a function
of $\gamma$ and $\delta$. In the darker region in the bottom half of the plot, this is 
where ruin probability is always $1$ regardless of the investment. In the upper half of the plot, the minimized
ruin probability is less than $1$ and it shows the heat map. Here, we take
$\rho=2$, $\nu=2$, and $\lambda=0.1$.}
\label{GammaDelta}
\end{center}
\end{figure}

In Figure~\ref{RhoLambda}, we investigate the dependence of the optimal $C^{\ast}$ on
the parameters $\rho$ and $\lambda$ given $\delta=1$, $\nu=0.1$ and $\gamma=\frac{1}{2}$.
For $\gamma=\frac{1}{2}$, we showed in Example~\ref{ExpExample} that
the optimal $C^{\ast}$ is given by
\begin{equation}
C^{\ast}=\frac{\delta^{2}\rho^{2}}{(\lambda+\sqrt{\lambda^{2}+\rho\delta^{2}})^{2}}.
\end{equation}
When $p(y)=\nu e^{-\nu y}$ and $\gamma=\frac{1}{2}$, the condition \eqref{ConditionOne} reduces to
$\rho-\frac{\lambda}{\nu}-\frac{\delta^{2}}{4\nu^{2}}<0$.
When this condition is violated, the ruin probability is always $1$
regardless of the investment and it corresponds to the dark region
in the right bottom corner of the plot in Figure~\ref{RhoLambda}. 
When this condition is satisfied, the heat map plot
of the optimal $C^{\ast}$ as a function of $\rho$ and $\lambda$
is illustrated in Figure~\ref{RhoLambda}. We can see
that as $\rho$ increases, the optimal $C^{\ast}$ increases, 
and as $\lambda$ increases, the optimal $C^{\ast}$ decreases.

\begin{figure}[htb]
\begin{center}
\includegraphics[scale=1.0]{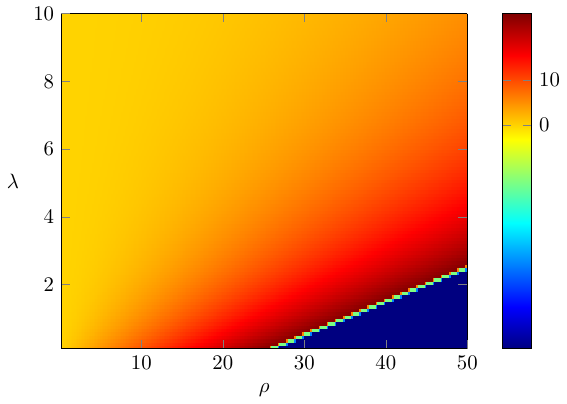}
\caption{
This shows $C^{\ast}$ as a function
of $\rho$ and $\lambda$. In the darker region in the right bottom corner of the plot, this is
where ruin probability is always $1$ regardless of the investment. In the rest of the plot, the minimized
ruin probability is less than $1$.
Here, we take $\nu=0.1$, $\gamma=0.5$ and $\delta=1$.}
\label{RhoLambda}
\end{center}
\end{figure}

\subsection{State-Dependent Ruin Probability with Optimal Investment}

In this section, we assume that the dual risk model is state-dependent, 
and in particular, we assume that
$F(x,c)=\lambda(x)+\delta(x)c^{\gamma}$.
We also assume that $Y_{i}$ are i.i.d. exponentially distributed so that $p(y)=\nu e^{-\nu y}$ for some $\nu>0$.

First, let us consider a special example in the case of $0<\gamma<1$.
Let us consider the model in Example~\ref{StateExI}.
For simplicity, let us assume that $\gamma=\frac{1}{2}$.
Recall that in Example~\ref{StateExI}, 
$\rho(x)=\rho_{0}$, $\lambda(x)=\lambda_{0}(c_{1}x+c_{2})$, 
and $\delta(x)=\delta_{0}(c_{1}x+c_{2})$. The optimal investment rate $C^{\ast}(x)\equiv C_{0}$ 
is a constant and is given by:
\begin{equation}
C_{0}=\frac{\delta_{0}^{2}\rho_{0}^{2}}{(\lambda_{0}+\sqrt{\lambda_{0}^{2}+\rho_{0}\delta_{0}^{2}})^{2}}.
\end{equation}
The minimized ruin probability is given by
\begin{equation}\label{minRuinState}
\frac{2a\sqrt{d}e^{cx-dx^{2}}+\sqrt{\pi}e^{\frac{c^{2}}{4d}}(ac+2bd)
\text{erfc}(\frac{2dx-c}{2\sqrt{d}})}{2a\sqrt{d}+\sqrt{\pi}e^{\frac{c^{2}}{4d}}(ac+2bd)\text{erfc}(\frac{-c}{2\sqrt{d}})},
\end{equation}
where $x$ is the initial wealth, $a:=c_{1}$, $b:=c_{2}$,
$c:=\nu-\frac{\lambda_{0}+\delta_{0}C_{0}^{1/2}}{\rho_{0}+C_{0}}c_{2}$,
and
$d:=\frac{\lambda_{0}+\delta_{0}C_{0}^{1/2}}{\rho_{0}+C_{0}}\frac{c_{1}}{2}$.
By setting $C_{0}=0$ in \eqref{minRuinState}, we get
the ruin probability without any investment in research and development.

In Figure~\ref{StateComparison}, the blue curve with circle markers stands for the ruin probability
without investment and the red dashed curve stands for the minimized ruin probability with investment.
These two curves differ from exponential decays, which is due to the flexibility
of the state-dependent model. As observed in \cite{ZhuState}, for state-dependent dual risk model, 
the ruin probability can have subexponential, exponential and superexpontial decays in terms
of the initial wealth. Also for the state-dependent dual risk model, the ruin probability may not
be convex in the initial wealth (as we can see from the blue curve with circle markers in Figure~\ref{StateComparison}).

\begin{figure}[htb]
\begin{center}
\includegraphics[scale=1.0]{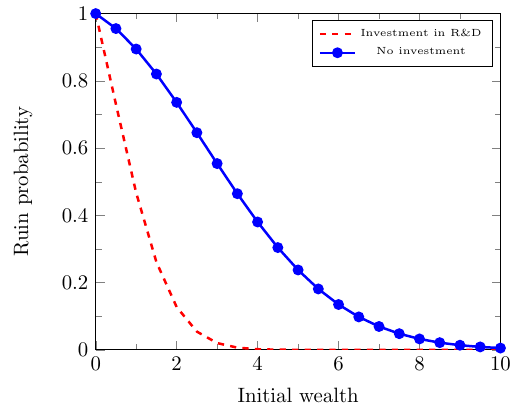}
\caption{Illustration of the ruin probability without any investment (blue curve with circle markers), the minimized ruin probability
with investment in research and development (red dashed curve). The $x$-axis denotes
the initial wealth of the underlying company and the $y$-axis denotes the (minimized) ruin probability.
Here, we take $\gamma=0.5$, $\rho_{0}=1$, $\nu=0.1$,  $\lambda_{0}=0.1$,
$\delta=1$, $c_{1}=1$, and $c_{2}=1$.}
\label{StateComparison}
\end{center}
\end{figure}

Next, let us consider an example for $\gamma=1$ for the state-dependent dual risk model.
Let us recall that in Example~\ref{StateExII}, 
$\rho(x)=\rho_{0}(c_{1}x+c_{2})$, $\lambda(x)=\left(\nu+\frac{\lambda_{0}}{1+x}\right)\rho(x)$, 
and $\delta(x)=\delta_{0}$, and under the assumption that
$\nu<\delta_{0}<\nu+\lambda_{0}$,
the optimal $C^{\ast}$ is given by $C^{\ast}=0$ if $x\leq x^{\ast}$ 
and $C^{\ast}=\infty$ if $x>x^{\ast}$, where
$x^{\ast}:=\frac{\lambda_{0}-\delta_{0}+\nu}{\delta_{0}-\nu}$.
From Example~\ref{StateExII}, with optimal investment, 
the minimized ruin probability is 
given by $V(x)$ in \eqref{greaterxast} if $x>x^{\ast}$
and the minimized ruin probability is given by $V(x)$ in \eqref{lessxast} if $x\leq x^{\ast}$,
where $x$ is the initial wealth. Without any investment, as in Theorem~\ref{MainThm2}, 
under the assumption that $\lambda_{0}>1$, we can compute
that the ruin probability is given by
\begin{align}
V(x)
=\frac{\int_{x}^{\infty}\left(\nu+\frac{\lambda_{0}}{1+y}\right)\frac{1}{(1+y)^{\lambda_{0}}}dy}
{\int_{0}^{\infty}\left(\nu+\frac{\lambda_{0}}{1+y}\right)\frac{1}{(1+y)^{\lambda_{0}}}dy}
=\frac{\nu(1+x)^{-\lambda_{0}+1}+(\lambda_{0}-1)(1+x)^{-\lambda_{0}}}{\lambda_{0}+\nu-1},
\end{align}
which is strictly between $0$ and $1$.
\begin{figure}[htb]
\begin{center}
\includegraphics[scale=1.0]{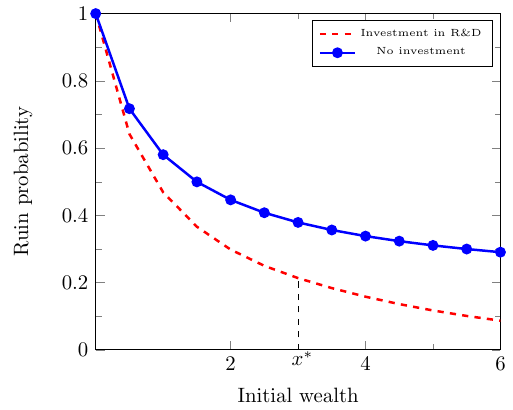}
\caption{Illustration of the ruin probability without any investment (blue curve with circle markers), the minimized ruin probability
with investment in research and development (red dashed curve). The $x$-axis denotes
the initial wealth of the underlying company and the $y$-axis denotes the (minimized) ruin probability.
$x^{\ast}$ on the $x$-axis is the critical threshold above which the optimal strategy is to invest
as much as possible in R\&D, and below which the optimal strategy is not to invest at all in R\&D.
Here, we take $\rho_{0}=1$ (irrelevant), $\nu=0.1$, $\lambda_{0}=1.2$, $\delta_{0}=0.4$, and $c_{1}=c_{2}=1$ (irrelevant) and $\gamma=1$.}
\label{StateSingularComparison}
\end{center}
\end{figure}
In Figure~\ref{StateSingularComparison}, we plot the curve of the ruin probability
as a function of the initial wealth without investment (blue curve with circle markers) and the minimized ruin probability
as a function of the initial wealth with the optimal investment in research and development (red dashed curve)
as in the example of the state-dependent dual risk model we described above. 
In Figure~\ref{StateSingularComparison}, the critical threshold for for the optimal investment
strategy is $x^{\ast}=3$ in the plot. When the wealth process
is below this threshold $x^{\ast}$, the optimal strategy for investment in R\&D is not to invest, 
and when the wealth process is above this threshold $x^{\ast}$, the optimal strategy
for investment in R\&D is to invest as aggressively as possible. 
When $x<x^{\ast}$, from \eqref{lessxast}, we can see that $V(x)$ decays polynomially
in $x$, and when $x>x^{\ast}$, from \eqref{greaterxast}, we can see that $V(x)$ decays
exponentially in $x$. 

\section*{Acknowledgements}
Arash Fahim gratefully acknowledges support from the National Science Foundation via the award
NSF-DMS-1447067. Lingjiong Zhu is grateful to the support from the National Science Foundation via the awards NSF-DMS-1613164, NSF-DMS-2053454, NSF-DMS-2208303.

\bibliographystyle{plain} 
\bibliography{DualRiskModel.bib}

\end{document}